%% file: ms.tex
\documentclass[11pt,a4paper,oneside,british,article]{memoir}

\usepackage[a-3u]{pdfx}
\settrims{0pt}{0pt}
\setlrmarginsandblock{1in}{*}{1}
\checkandfixthelayout

\usepackage[utf8]{inputenc}
\usepackage[T1]{fontenc}
\usepackage{amsmath}
\usepackage{amssymb}
\usepackage{amsthm}
\usepackage{mathtools}
\usepackage[british, UKenglish]{babel}
\usepackage{microtype}
\usepackage{rotating}
\usepackage{stackengine}
\DeclareUnicodeCharacter{1EC5}{\stackon[-5.5pt]{ê}{\~{}}}
\usepackage[natbib, doi, backref, useprefix, citestyle=alphabetic, bibstyle=alphabetic, maxbibnames=20, minbibnames=5]{biblatex}
\DefineBibliographyExtras{british}{}
\usepackage{newpxtext,newpxmath}
\usepackage{cleveref}

\epigraphfontsize{\small\itshape}
\usepackage{adforn}
\newcommand*{\sectionbreak}{\fancybreak{\adfflourishleftdouble\quad\adfast9\quad\adfflourishrightdouble}}

\newcommand{\R}{\mathbb{R}}

\newcommand{\Npos}{\mathbb{N}_1}
\newcommand{\K}{\mathbb{K}}
\newcommand{\CC}{\mathbb{C}}

\newcommand{\T}{\transp}
\newcommand{\bigO}{\mathcal{O}}
\newcommand{\eps}{\varepsilon}
\newcommand{\zerovector}{\ensuremath{\mathbf{0}}}

\providecommand{\eqdef}{\coloneq}

\DeclareMathOperator{\spn}{span}
\DeclareMathOperator{\tail}{tail}
\newcommand{\dotprod}[3][]{\mathopen#1\langle #2 , #3 \mathclose#1\rangle}
\newcommand{\abs}[2][]{\mathopen#1\lvert #2 \mathclose#1\rvert}
\newcommand{\norm}[3][]{\mathopen#1\lVert #2 \mathclose#1\rVert_#3}

\DeclareMathOperator*{\argmin}{arg\,min}
\DeclareMathOperator*{\median}{median}

\makeatletter
\newtheorem*{rep@theorem}{\rep@title}
\newcommand{\newreptheorem}[2]{%
\newenvironment{rep#1}[1]{%
 \def\rep@title{#2 \ref{##1}}%
 \begin{rep@theorem}}%
 {\end{rep@theorem}}}
\makeatother

\newtheorem{theorem}{Theorem}[chapter]
\newreptheorem{theorem}{Theorem}
\newtheorem{lemma}[theorem]{Lemma}
\newreptheorem{lemma}{Lemma}

\newtheorem{definition}[theorem]{Definition}

\newtheorem{proposition}[theorem]{Proposition}
\newreptheorem{proposition}{Proposition}
\newtheorem{corollary}[theorem]{Corollary}
\newreptheorem{corollary}{Corollary}

\newreptheorem{conjecture}{Conjecture}

\setsecnumdepth{subsection}

\begin{document}

\title{An Introduction to Johnson--Lindenstrauss Transforms}
\author{Casper Benjamin Freksen\thanks{Work done while at the Computer
    Science Department, Aarhus University.
    \href{mailto:casper@freksen.dk}{\texttt{casper@freksen.dk}}}}
\date{\today}

\input{frontmatter}

\input{introduction}

\input{appendices}

\printbibliography

\end{document}

%% file: frontmatter.tex
\maketitle

\begin{abstract}
  Johnson--Lindenstrauss Transforms are powerful tools for reducing
  the dimensionality of data while preserving key characteristics of
  that data, and they have found use in many fields from machine
  learning to differential privacy and more. This note explains what
  they are; it gives an overview of their use and their development
  since they were introduced in the 1980s; and it provides many
  references should the reader wish to explore these topics more
  deeply.

  The text was previously a main part of the introduction of my PhD
  thesis~\citep{Freksen:2020:aSoJaL}, but it has been adapted to
  be self contained and serve as a (hopefully good) starting point for
  readers interested in the topic.
\end{abstract}


%% file: introduction.tex
\chapter{The Why, What, and How}
\label{cha:introduction}

\section{The Problem}
\label{sec:intro:problem}

Consider the following scenario: We have some data that we wish to
process but the data is too large, e.g.\@ processing the data takes
too much time, or storing the data takes too much space. A solution
would be to compress the data such that the valuable parts of the data
are kept and the other parts discarded. Of course, what is considered
valuable is defined by the data processing we wish to apply to our
data. To make our scenario more concrete let us say that our data
consists of vectors in a high dimensional Euclidean space, $\R^d$, and
we wish to find a transform to embed these vectors into a lower
dimensional space, $\R^m$, where $m \ll d$, so that we can apply our
data processing in this lower dimensional space and still get
meaningful results. The problem in this more concrete scenario is
known as \emph{dimensionality reduction}.

As an example, let us pretend to be a library with a large corpus of
texts and whenever a person returns a novel that they liked, we would
like to recommend some similar novels for them to read next. Or
perhaps we wish to be able to automatically categorise the texts into
groups such as fiction/non-fiction or child/young-adult/adult
literature. To be able to use the wealth of research on similarity
search and classification we need a suitable representation of our
texts, and here a common choice is called bag-of-words. For a language
or vocabulary with $d$ different words, the bag-of-words
representation of a text $t$ is a vector $x \in \R^d$ whose $i$th
entry is the number of times the $i$th word occurs in $t$. For
example, if the language is just [``be'', ``is'', ``not'', ``or'',
``question'', ``that'', ``the'', ``to''] then the text ``to be or not
to be'' is represented as $(2, 0, 1, 1, 0, 0, 0, 2)^\T$. To capture
some of the context of the words, we can instead represent a text as
the count of so-called $n$-grams\footnote{These are sometimes referred
  to as shingles.}, which are sequences of $n$ consecutive words,
e.g.\@ the 2-grams of ``to be or not to be'' are [``to be'', ``be
or'', ``or not'', ``not to'', ``to be''], and we represent such a
bag-of-$n$-grams as a vector in $\R^{(d^n)}$. To compare two texts we
compute the distance between the vectors of those texts, because the
distance between vectors of texts with mostly the same words (or
$n$-grams) is small\footnote{We might wish to apply some normalisation
  to the vectors, e.g.\@ tf-idf~\citep{Leskovec:2020:MoMD:1}, so that
  rare words are weighted more and text length is less
  significant.}. For a more realistic language such as English with
$d \approx 171 000$ words~\citep{Simpson:1989:OED} or that of the
``English'' speaking internet at $d \gtrsim 4 790 000$
words~\citep{Williams:2005:SWotW}, the dimension quickly becomes
infeasable. While we only need to store the nonzero counts of words
(or $n$-grams) to represent a vector, many data processing algorithms
have a dependency on the vector dimension $d$ (or $d^n$), e.g.\@ using
nearest-neighbour search to find similar novels to
recommend~\citep{Andoni:2017:NNiHDS} or using neural networks to
classify our texts~\citep{Schmidt:2018:SRPLGPDRfDL}. These algorithms
would be infeasible for our library use case if we do not first reduce
the dimension of our data.

A seemingly simple approach would be to select a subset of the
coordinates, say if the data contained redundant or irrelevant
coordinates. This is known as \emph{feature
  selection}~\citep{Jensen:2008:CIaFSRaFA, Hastie:2017:tEoSLDMIaP,
  James:2013:aItSL:6}, and can be seen as projecting\footnote{Here we
  slightly bend the definition of projection in the sense that we
  represent a projected vector as the coefficients of the of the
  linear combination of the (chosen) basis of the subspace we project
  onto, rather than as the result of that linear combination. If
  $A \in \R^{m \times d}$ is a matrix with the subspace basis vectors
  as rows, we represent the projection of a vector $x \in \R^d$ as the
  result of $Ax \in \R^m$ rather than $A^\T A x \in \R^d$.} onto an
axis aligned subspace, i.e.\@ a subspace whose basis is a subset of
$\{e_1, \dotsc, e_d\}$.

We can build upon \emph{feature selection} by choosing the basis from
a richer set of vectors. For instance, in principal component analysis
as dimensionality reduction (PCA)~\citep{Pearson:1901:oLaPoCFtSoPiS,
  Hotelling:1933:AoaCoSViPC} we let the basis of the subspace be the
$m$ first eigenvectors (ordered decreasingly by eigenvalue) of
$X^\T X$, where the rows of $X \in \R^{n \times d}$ are our $n$ high
dimensional vectors\footnote{Here it is assumed that the mean of our
  vectors is $\zerovector$, otherwise the mean vector of our vectors
  should be subtracted from each of the rows of $X$.}. This subspace
maximises the variance of the data in the sense that the first
eigenvector is the axis that maximises variance and subsequent
eigenvectors are the axes that maximise variance subject to being
orthogonal to all previous eigenvectors~\citep{Leskovec:2020:MoMD:11,
  Hastie:2017:tEoSLDMIaP, James:2013:aItSL:10}.

But what happens if we choose a basis randomly?

\section{The Johnson--Lindenstrauss Lemma(s)}
\label{sec:intro:jl_lemma}

In \citeyear{Johnson:1984:EoLMiaHS}\footnote{Or rather in 1982 as that
  was when a particular ``Conference in Modern Analysis and
  Probability'' was held at Yale University, but the proceedings were
  published in \citeyear{Johnson:1984:EoLMiaHS}.} it was discovered
that projecting onto a random basis approximately preserves pairwise
distances with high probability. In order to prove a theorem regarding
Lipschitz extensions of functions from metric spaces into $\ell_2$,
\citet{Johnson:1984:EoLMiaHS} proved the following lemma.

\begin{lemma}[Johnson--Lindenstrauss lemma~\citep{Johnson:1984:EoLMiaHS}]
  \label{thm:jl}
  For every $d \in \Npos, \eps \in (0, 1)$, and $X \subset \R^d$,
  there exists a function $f \colon X \to \R^m$, where
  $m = \Theta(\eps^{-2} \log |X|)$ such that for every $x, y \in X$,
  \begin{equation}
    \label{eq:jl}
    \bigl|\, \|f(x) - f(y)\|_2^2 - \|x - y\|_2^2 \,\bigr| \leq \eps \|x - y\|_2^2.
  \end{equation}
\end{lemma}

\begin{proof}
  The gist of the proof is to first define
  $f(x) \eqdef (d/m)^{1/2}Ax$, where $A \in \R^{m \times d}$ are the
  first $m$ rows of a random orthogonal matrix. They then showed that
  $f$ preserves the norm of any vector with high probability, or more
  formally that the distribution of $f$ satisfies the following lemma.

  \begin{lemma}[Distributional Johnson--Lindenstrauss lemma~\citep{Johnson:1984:EoLMiaHS}]
    \label{thm:distjl}
    For every $d \in \Npos$ and $\eps, \delta \in (0, 1)$, there
    exists a probability distribution $\mathcal{F}$ over linear
    functions $f \colon \R^d \to \R^m$, where
    $m = \Theta(\eps^{-2}\log\frac{1}{\delta})$ such that for every
    $x \in \R^d$,
    \begin{equation}
      \label{eq:distjl}
      \Pr_{f \sim \mathcal{F}}\Bigl[ \, \bigl| \,\|f(x)\|_2^2 - \|x\|_2^2\, \bigr| \leq \eps \|x\|_2^2 \,\Bigr] \geq 1 - \delta .
    \end{equation}
  \end{lemma}

  By choosing $\delta = 1/|X|^2$, we can union bound over all pairs of
  vectors $x, y \in X$ and show that their distance (i.e.\@ the
  $\ell_2$ norm of the vector $x - y$) is preserved simultaneously for
  all pairs with probability at least
  $1 - \binom{|X|}{2}/|X|^2 > 1/2$.
\end{proof}

We will use the term Johnson--Lindenstrauss distribution (JLD) to
refer to a distribution $\mathcal{F}$ that is a witness to
\cref{thm:distjl}, and the term Johnson--Lindenstrauss transform (JLT)
to a function $f$ witnessing \cref{thm:jl}, e.g.\@ a sample of a JLD.

A few things to note about these lemmas are that when sampling a JLT
from a JLD it is independent of the input vectors themselves; the JLT
is only dependendent on the source dimension $d$, number of vectors
$|X|$, and distortion $\eps$. This allows us to sample a JLT without
having access to the input data, e.g.\@ to compute the JLT before the
data exists, or to compute the JLT in settings where the data is too
large to store on or move to a single machine\footnote{Note that the
  sampled transform only satisfies \cref{thm:jl} with some (high)
  probability. In the setting where we have access to the data, we can
  avoid this by resampling the transform until it satisfies
  \cref{thm:jl} for our specific dataset.}. Secondly, the target
dimension $m$ is independent from the source dimension $d$, meaning
there are potentially very significant savings in terms of
dimensionality, which will become more apparent shortly.

Compared to PCA, the guarantees that JLTs give are different: PCA
finds an embedding with optimal average distortion of distances
between the original and the embedded vectors, i.e.\@
$A_{\mathsf{PCA}} = \argmin_{A \in \R^{m \times d}} \sum_{x \in
  X}\norm{A^\T A x - x}{2}^2$~\citep{Jolliffe:2002:PCA}, whereas a JLT
bounds the worst case distortion between the distances within the
original space and distances within the embedded space. As for
computing the transformations, a common\footnote{There are also other
  approaches that compute an approximate PCA more efficiently than
  this, e.g.~\citep{Rokhlin:2009:aRAfPCA, Anaraki:2014:MaCEPCAvVSRP}.}
way of performing PCA is done by computing the covariance matrix and
then performing eigenvalue decomposition, which results in a running
time\footnote{Here $\omega \lesssim 2.373$ is the exponent from the
  running time of squared matrix
  multiplication~\citep{Williams:2012:MMFtCW, LeGall:2014:PoTaFMM}.}
of $\bigO(\abs{X}d^2 + d^\omega)$~\citep{Demmel:2007:FLAiS}, compared
to $\Theta(\abs{X} d \log d)$ and\footnote{Here $\norm{X}{0}$ is the
  total number of nonzero entries in the set of vectors $X$, i.e.\@
  $\norm{X}{0} \eqdef \sum_{x \in X} \norm{x}{0}$ where
  $\norm{x}{0} \eqdef \abs{\{i \mid x_i \neq 0\}}$ for a vector $x$.}
$\Theta(\norm{X}{0} \eps^{-1} \log \abs{X})$ that can be achieved by
the JLDs ``FJLT'' and ``Block SparseJL'', respectively, which will be
introduced in \cref{sec:intro:evolution-jl}. As such, PCA and JLDs are
different tools appropriate for different scenarios (see e.g.\@
\citep{Breger:2019:otRAoMCMRIwOP} where the two techniques are
compared empirically in the domain of medicinal imaging; see
also~\citep{Dasgupta:2000:EwRP, Bingham:2001:RPiDRAtIaTD,
  Fern:2003:RPfHDDCaCEA, Fradkin:2003:EwRPfML, Tang:2005:CDRTfDC,
  Deegalla:2006:RHDDbPCAvRPfNNC, Arpit:2014:aAoRPiCB,
  Wojnowicz:2016:PBtRhtRtDoVLDiaWtORP,
  Breger:2020:oOPfDRaAiATLFfLP}). That is not to say the two are
mutually exclusive, as one could apply JL to quickly shave off some
dimensions followed by PCA to more carefully reduce the remaining
dimensions~\citep[e.g.\@][]{Rokhlin:2009:aRAfPCA,
  Halko:2011:FSwRPAfCAMD, Xie:2016:CaDRTBoRPfCC,
  Yang:2020:htRDwPCAaRP}. For more on PCA, we refer the interested
reader to \citep{Jolliffe:2002:PCA}, which provides an excellent in
depth treatment of the topic.

One natural question to ask with respect to JLDs and JLTs is if the
target dimension is optimal. This is indeed the case as
\citet{Kane:2011:AOEJLF,Jayram:2013:OBfJLTaSPwSCE} independently give
a matching lower bound of
$m = \Omega(\eps^{-2} \log \frac{1}{\delta})$ for any JLD that
satisfies \cref{thm:distjl}, and \citet{Larsen:2017:OotJLL} showed
that the bound in \cref{thm:jl} is optimal up constant factors for
almost the entire range of $\eps$ with the following theorem.
\begin{theorem}[\citep{Larsen:2017:OotJLL}]
  For any integers $n, d \geq 2$ and
  $\lg^{0.5001} n / \sqrt{\min \{d, n\}} < \eps < 1$ there exists a
  set of points $X \subset \R^d$ of size $n$ such that any function
  $f : X \to \R^m$ satisfying \cref{eq:jl} must have
  \begin{equation}
    m = \Omega\bigl(\eps^{-2} \log (\eps^2 n)\bigr).
  \end{equation}
\end{theorem}
Note that if $\eps \leq \sqrt{\lg n / \min \{d, n\}}$ then
$\eps^{-2} \lg n \geq \min\{d, n\}$, and embedding $X$ into dimension
$\min \{d, |X|\}$ can be done isometrically by the identity function
or by projecting onto $\spn(X)$, respectively.

\citet{Alon:2017:OCoAIPaDR} extended the result in
\citep{Larsen:2017:OotJLL} by providing a lower bound for the gap in
the range of $\eps$.
\begin{theorem}[\citep{Alon:2017:OCoAIPaDR}]
  There exists an absolute positive constant $0 < c < 1$ so that for any
  $n \geq d > cd \geq m$ and for all $\eps \geq 2/\sqrt{n}$, there
  exists a set of points $X \subset \R^d$ of size $n$ such that any
  function $f : X \to \R^m$ satisfying \cref{eq:jl} must have
  \begin{equation}
    m = \Omega\bigl(\eps^{-2} \log (2 + \eps^2 n) \bigr).
  \end{equation}
\end{theorem}

It is, however, possible to circumvent these lower bounds by
restricting the set of input vectors we apply the JLTs to. For
instance, \citet{Klartag:2005:EPaRP, Dirksen:2016:DRwSMaUT,
  Bourgain:2015:TaUToSDRiES} provide target dimension upper bounds for
JLTs that are dependent on statistical properties of the input set
$X$. Similarly, JLTs can be used to approximately preserve pairwise
distances simultaneously for an entire subspace using
$m = \Theta(\eps^{-2} t \log (t/\eps))$, where $t$ denotes the
dimension of the subspace~\citep{Sarlos:2006:IAAfLMvRP}, which is a
great improvement when $t \ll |X|, d$.

Another useful property of JLTs is that they approximately preserve
dot products. \Cref{thm:jl:dotprod} formalises this property in terms
of \cref{thm:jl}, though it is
sometimes~\citep{Sarlos:2006:IAAfLMvRP, Arriaga:2006:aAToLRCaRP} stated
in terms of \cref{thm:distjl}. \Cref{thm:jl:dotprod} has a few extra
requirements on $f$ and $X$ compared to \cref{thm:jl}, but these are
not an issue if the JLT is sampled from a JLD, or if we add the
negations of all our vectors to $X$, which only slightly increases the
target dimension.

\begin{corollary}
  \label{thm:jl:dotprod}
  Let $d, \eps, X$ and $f$ be as defined in \cref{thm:jl}, and
  furthermore let $f$ be linear. Then for every $x, y \in X$, if
  $-y \in X$ then
  \begin{equation}
    \label{eq:jl:dotprod}
    | \dotprod{f(x)}{f(y)} - \dotprod{x}{y} | \leq \eps \|x\|_2 \|y\|_2.
  \end{equation}
\end{corollary}
\begin{proof}
  If at least one of $x$ and $y$ is the \zerovector-vector, then
  \cref{eq:jl:dotprod} is trivially satisfied as $f$ is linear. If $x$
  and $y$ are both unit vectors then we assume w.l.o.g.\@ that
  $\norm{x + y}{2} \geq \norm{x - y}{2}$ and we proceed as follows,
  utilising the polarisation identity:
  $4 \dotprod{u}{v} = \|u + v\|_2^2 - \|u - v\|_2^2$.
  \begin{align*}
    4 \bigl| \dotprod{f(x)}{f(y)} - \dotprod{x}{y} \bigr|
    &= \bigl| \|f(x) + f(y)\|_2^2 - \|f(x) - f(y)\|_2^2 - 4 \dotprod{x}{y} \bigr| \\
    &\leq \bigl| (1 + \eps)\|x + y\|_2^2 - (1 - \eps)\|x - y\|_2^2 - 4 \dotprod{x}{y} \bigr| \\
    &= \bigl| 4 \dotprod{x}{y} + \eps(\|x + y\|_2^2 + \|x - y\|_2^2) - 4 \dotprod{x}{y} \bigr| \\
    &= \eps(2\|x\|_2^2 + 2\|y\|_2^2) \\
    &= 4 \eps .
  \end{align*}

  Otherwise we can reduce to the unit vector case.
  \begin{align*}
    | \dotprod{f(x)}{f(y)} - \dotprod{x}{y} |
    &= \biggl| \dotprod[\Big]{f\bigl(\frac{x}{\|x\|_2}\bigr)}{f\bigl(\frac{y}{\|y\|_2}\bigr)} - \dotprod[\Big]{\frac{x}{\|x\|_2}}{\frac{y}{\|y\|_2}} \biggr| \|x\|_2\|y\|_2 \\
    &\leq \eps \|x\|_2\|y\|_2 .
  \end{align*}
\end{proof}

\sectionbreak

Before giving an overview of the development of JLDs in
\cref{sec:intro:evolution-jl}, let us return to our scenario and
example in \cref{sec:intro:problem} and show the wide variety of
fields where dimensionality reduction via JLTs have found use.
Furthermore, to make us more familiar with \cref{thm:jl} and its
related concepts, we will pick a few examples of how the lemma is
used.

\section{The Use(fulness) of Johnson--Lindenstrauss}
\label{sec:intro:jl-use-cases}

JLDs and JLTs have found uses and parallels in many fields and tasks,
some of which we will list below. Note that there are some overlap
between the following categories, as e.g.\@
\citep{Fern:2003:RPfHDDCaCEA} uses a JLD for an \emph{ensemble} of
weak learners to learn a mixture of Gaussians \emph{clustering}, and
\citep{Pilanci:2015:RSoCPwSG} solves \emph{convex optimisation}
problems in a way that gives \emph{differential privacy} guarantees.

\begin{description}

\item[Nearest-neighbour search] have benefited from the
  Johnson--Lindenstrauss lemmas on multiple occasions, including
  \citep{Kleinberg:1997:TAfNNSiHD, Kushilevitz:2000:ESfANNiHDS}, which
  used JL to randomly partition space rather than reduce the
  dimension, while others~\citep{Ailon:2009:tFJLTaANN,
    Har-Peled:2012:ANNTRtCoD} used the dimensionality reduction
  properties of JL more directly. Variations on these results include
  consructing locality sensitive hashing
  schemes~\citep{Datar:2004:LSHSBopSD} and finding nearest neighbours
  without false negatives~\citep{Sankowski:2017:ANNSwFNflfc}.

\item[Clustering] with results in various sub-areas such as mixture of
  Gaussians~\citep{Dasgupta:1999:LMoG, Fern:2003:RPfHDDCaCEA,
    Urruty:2007:CbRP}, subspace clustering~\citep{Heckel:2017:DRSC},
  graph clustering~\citep{Sakai:2009:FSCwRPaS, Guo:2020:RScCfLSDN},
  self-organising maps~\citep{Ritter:1989:SOSM,
    Kaski:1998:DRbRMFSCfC}, and
  $k$-means~\citep{Becchetti:2019:ODRfkMBSatJLL,
    Cohen:2015:DRfkMCaLRA, Boutsidis:2014:RDRfkMC,
    Liu:2017:FCSCaCEwRP, Sieranoja:2018:RPfkMC}, which will be
  explained in more detail in~\cref{sec:intro:clustering}.

\item[Outlier detection] where there have been works for various
  settings of outliers, including approximate
  nearest-neighbours~\citep{deVries:2010:FLAiVHDS,
    Schubert:2015:FaSODwANNE} and Gaussian
  vectors~\citep{Navarro-Esteban:2020:HDODuRP}, while
  \citep{Zhao:2020:SUODtSUOD} uses JL as a preprocessor for a range of
  outlier detection algorithms in a distributed computational model,
  and \citep{Aouf:2012:ADODuRSP} evaluates the use of JLTs for outlier
  detection of text documents.

\item[Ensemble learning] where independent JLTs can be used to
  generate training sets for weak learners for
  bagging~\citep{Schclar:2009:RPEC} and with the voting among the
  learners weighted by how well a given JLT projects the
  data~\citep{Cannings:2017:RPEC, Cannings:2020:RPDPfCP}. The
  combination of JLTs with multiple learners have also found use in
  the regime of learning high-dimensional distributions from few
  datapoints (i.e.\@
  $\abs{X} \ll d$)~\citep{Durrant:2013:RPaRLaLDEfFOtD,
    Zhang:2019:EwRPELvQD, Niyazi:2020:AAoaEoRPLD}.

\item[Adversarial machine learning] where Johnson--Lindenstrauss can
  both be used to defend against adversarial
  input~\citep{Nguyen:2016:TRMuRP, Weerasinghe:2019:SVMRaTDIA,
    Taran:2019:DaAAbRD} as well as help craft such
  attacks~\citep{Li:2020:PaPDBBA}.

\item[Miscellaneous machine learning] where, in addition to the more
  specific machine learning topics mentioned above,
  Johnson--Lindenstrauss has been used together with support vector
  machines~\citep{Calderbank:2009:CLUSDRaLitMD, Paul:2014:RPfLSVM,
    Lei:2020:ISRHTfLSVM}, Fisher's linear
  discriminant~\citep{Durrant:2010:CFLDACoRPD}, and neural
  networks~\citep{Schmidt:2018:SRPLGPDRfDL}, while
  \citep{Kim:2020:ARPGwVR} uses JL to facilitate stochastic gradient
  descent in a distributed setting.

\item[Numerical linear algebra] with work focusing on low rank
  approximation~\citep{Cohen:2015:DRfkMCaLRA, Musco:2020:PCPSPSaC},
  canonical correlation analysis~\citep{Avron:2014:EDRfCCA}, and
  regression in a local~\citep{Thanei:2017:RPfLSR, Maillard:2009:CLSR,
    Kaban:2014:NBoCLLSR, Slawski:2017:CLSQR} and a
  distributed~\citep{Heinze:2016:DLDSEuRP} computational
  model. Futhermore, as many of these subfields are related some
  papers tackle multiple numerical linear algebra problems at once,
  e.g.\@ low rank approximation, regression, and approximate matrix
  multiplication~\citep{Sarlos:2006:IAAfLMvRP}, and a line of
  work~\citep{Meng:2013:LDSEiISTaAtRLR, Clarkson:2017:LRAaRiIST,
    Nelson:2013:OFNLAAvSSE} have used JLDs to perform subspace
  embeddings which in turn gives algorithms for $\ell_p$ regression,
  low rank approximation, and leverage scores.

  For further reading, there are surveys~\citep{Mahoney:2011:RAfMaD,
    Halko:2011:FSwRPAfCAMD, Woodruff:2014:SaaTfNLA} covering much of
  JLDs' use in numerical linear algebra.

\item[Convex optimisation] in which Johnson--Lindenstrauss has been
  used for (integer) linear programming~\citep{Vu:2015:UtJLLiLaIP} and
  to improve a cutting plane method for finding a point in a convex
  set using a separation oracle~\citep{TatLee:2015:aFCPMaiIfCaCO,
    Jiang:2020:aICPMfCOCCGaiA}. Additionally,
  \citep{Zhang:2013:RtOSbDRP} studies how to recover a
  high-dimensional optimisation solution from a JL dimensionality
  reduced one.

\item[Differential privacy] have utilised Johnson--Lindenstrauss to
  provide sanitised solutions to the linear algebra problems of
  variance estimation~\citep{Blocki:2012:tJLTIPDP},
  regression~\citep{Sheffet:2019:DPOLS,
    Showkatbakhsh:2019:PUTOoLRuRPaAN, Zhou:2009:CaPSSR}, Euclidean
  distance estimation~\citep{Kenthapadi:2013:PvtJLT,
    Liu:2006:RPBMDPfPPDDM, Giannella:2013:BEDPDPufKI,
    Turgay:2008:DRoDPDT, Xu:2017:DPProDPHDDRvRP}, and low-rank
  factorisation~\citep{Upadhyay:2018:tPoPfLRF}, as well as convex
  optimisation~\citep{Pilanci:2015:RSoCPwSG,
    Kasiviswanathan:2016:EPERMfHDL}, collaborative
  filtering~\citep{Yang:2017:PPCFvtJLT} and solutions to graph-cut
  queries~\citep{Blocki:2012:tJLTIPDP,
    Upadhyay:2013:RPGSaDP}. Furthermore,
  \citep{Upadhyay:2015:REFJLTwAiDPaCS} analysis various JLDs with
  respect to differential privacy and introduces a novel one designed
  for this purpose.

\item[Neuroscience] where it is used as a tool to process data in
  computational neuroscience~\citep{Ganguli:2012:CSSaDiNIPaDA,
    Advani:2013:SMoCNSaHDD}, but also as a way of modelling
  neurological processes~\citep{Ganguli:2012:CSSaDiNIPaDA,
    Advani:2013:SMoCNSaHDD, Allen-Zhu:2014:SSCJLMCwNBC,
    Petrantonakis:2014:aCSPoHF}. Interestingly, there is some
  evidence~\citep{Murthy:2008:TORSitDMB, Stettler:2009:RoOitPC,
    Caron:2013:RCoOIitDMB} to suggest that JL-like operations occur in
  nature, as a large set of olifactory sensory inputs (projection
  neurons) map onto a smaller set of neurons (Kenyon cells) in the
  brains of fruit flies, where each Kenyon cell is connected to a
  small and seemingly random subset of the projection neurons. This is
  reminiscent of sparse JL constructions, which will be introduced in
  \cref{sec:intro:sparse-jl}, though I am not neuroscientifically
  adept enough to judge how far these similarities between biological
  constructs and randomised linear algebra extend.

\item[Other topics] where Johnson--Lindenstrauss have found use
  include graph sparsification~\citep{Spielman:2011:GSbER}, graph
  embeddings in Euclidean spaces~\citep{Frankl:1988:tJLLatSosG},
  integrated circuit design~\citep{Vempala:1998:RPanAtVLSI}, biometric
  authentication~\citep{Arpit:2014:aAoRPiCB}, and approximating
  minimum spanning trees~\citep{Har-Peled:2000:WCCAtMST}.

  For further examples of Johnson--Lindenstrauss use cases, please
  see~\citep{Indyk:2001:AAoLDGE, Vempala:2004:tRPM}.
\end{description}

Now, let us dive deeper into the areas of clustering and streaming
algorithms to see how Johnson--Lindenstrauss can be used there.

\subsection{Clustering}
\label{sec:intro:clustering}

Clustering can be defined as partitioning a dataset such that elements
are similar to elements in the same partition while being dissimilar
to elements in other partitions. A classic clustering problem is the
so-called $k$-means clustering where the dataset $X \subset \R^d$
consists of points in Euclidean space. The task is to choose $k$
cluster centers $c_1, \dotsc, c_k$ such that they minimise the sum of
squared distances from datapoints to their nearest cluster center,
i.e.\@
\begin{equation}
  \label{eq:kmeans:def_argmin}
  \argmin_{c_1, \dotsc, c_k} \sum_{x \in X} \min_i \|x - c_i\|_2^2.
\end{equation}

This creates a Voronoi partition, as each datapoint is assigned to the
partition corresponding to its nearest cluster center. We let
$X_i \subseteq X$ denote the set of points that have $c_i$ as their
closest center. It is well known that for an optimal choice of
centers, the centers are the means of their corresponding partitions,
and furthermore, the cost of any choice of centers is never lower than
the sum of squared distances from datapoints to the mean of their
assigned partition, i.e.\@
\begin{equation}
  \sum_{x \in X} \min_i \|x - c_i\|_2^2
  \geq \sum_{i = 1}^{k} \sum_{x \in X_i} \Bigl\|x - \frac{1}{|X_i|}\sum_{y \in X_i} y \Bigr\|_2^2.
\end{equation}

It has been shown that finding the optimal centers, even for $k = 2$,
is NP-hard~\citep{Aloise:2009:NPHoESoSC,Dasgupta:2008:tHokMC}; however, various heuristic
approaches have found success such as the commonly used Lloyd's
algorithm~\citep{Lloyd:1982:LSQiPCM}. In Lloyd's algorithm, after
initialising the centers in some way we iteratively improve the choice
of centers by assigning each datapoint to its nearest center and then
updating the center to be the mean of the datapoints assigned to
it. These two steps can then be repeated until some termination
criterion is met, e.g.\@ when the centers have converged. If we let $t$
denote the number of iterations, then the running time becomes
$\bigO(t |X| k d)$, as we use $\bigO(|X| k d)$ time per iteration to
assign each data point to its nearest center. We can improve this
running time by quickly embedding the datapoints into a lower
dimensional space using a JLT and then running Lloyd's algorithm in
this smaller space. The Fast Johnson--Lindenstrauss Transform, which
we will introduce later, can for many sets of parameters embed a
vector in $\bigO(d \log d)$ time reducing the total running time to
$\bigO(|X| d \log d + t |X| k \eps^{-2} \log |X|)$. However, for this
to be useful we need the partitioning of points in the lower
dimensional space to correspond to an (almost) equally good partition
in the original higher dimensional space.

In order to prove such a result we will use the following lemma, which
shows that the cost of a partitioning, with its centers chosen as the
means of the partitions, can be written in terms of pairwise distances
between datapoints in the partitions.
\begin{lemma}
  Let $k, d \in \Npos$ and $X_i \subset \R^d$ for\footnote{We use $[k]$
    to denote the set $\{1, \dotsc, k\}$.} $i \in [k].$
  \label{thm:kmeans:cost_pwd}
  \begin{equation}
    \label{eq:kmeans:cost_pwd}
    \sum_{i = 1}^{k} \sum_{x \in X_i} \Bigl\|x - \frac{1}{|X_i|}\sum_{y \in X_i} y \Bigr\|_2^2
    = \frac{1}{2} \sum_{i = 1}^{k} \frac{1}{|X_i|} \sum_{x, y \in X_i}\|x - y\|_2^2.
  \end{equation}
\end{lemma}
The proof of \cref{thm:kmeans:cost_pwd} consists of various linear
algebra manipulations and can be found in
\cref{sec:app:k-means-cost}. Now we are ready to prove the
following proposition, which states that if we find a partitioning
whose cost is within $(1 + \gamma)$ of the optimal cost in low
dimensional space, that partitioning when moving back to the high
dimensional space is within $(1 + 4 \eps) (1 + \gamma)$ of the optimal cost
there.

\begin{proposition}
  \label{thm:kmeans:cost_jl}
  Let $k, d \in \Npos$, $X \subset \R^d$, $\eps \leq 1/2$,
  $m = \Theta(\eps^{-2} \log |X|)$, and $f : X \to \R^m$ be a JLT. Let
  $Y \subset \R^m$ be the embedding of $X$. Let $\kappa_m^*$ denote
  the optimal cost of a partitioning of $Y$, with respect to
  \cref{eq:kmeans:def_argmin}. Let $Y_1, \dotsc, Y_k \subseteq Y$ be a
  partitioning of $Y$ with cost $\kappa_m$ such that
  $\kappa_m \leq (1 + \gamma)\kappa_m^*$ for some $\gamma \in \R$. Let
  $\kappa_d^*$ be the cost of an optimal partitioning of $X$ and
  $\kappa_d$ be the cost of the partitioning
  $X_1, \dotsc, X_k \subseteq X$, satisfying
  $Y_i = \{f(x) \mid x \in X_i\}$. Then
  \begin{equation}
    \label{eq:kmeans:cost_jl}
    \kappa_d \leq (1 + 4 \eps)(1 + \gamma) \kappa_d^*.
  \end{equation}
\end{proposition}
\begin{proof}
  Due to \cref{thm:kmeans:cost_pwd} and the fact that $f$ is a JLT we
  know that the cost of our partitioning is approximately preserved
  when going back to the high dimensional space, i.e.\@
  $\kappa_d \leq \kappa_m/(1 - \eps)$. Furthermore, since the cost
  of $X$'s optimal partitioning when embedded down to $Y$ cannot be
  lower than the optimal cost of partitioning $Y$, we can conclude
  $\kappa_m^* \leq (1 + \eps) \kappa_d^*$. Since $\eps \leq 1/2$, we have
  $1/(1 - \eps) = 1 + \eps/(1 - \eps) \leq 1 + 2 \eps$ and also
  $(1 + \eps)(1 + 2 \eps) = (1 + 3 \eps + 2 \eps^2) \leq 1 + 4 \eps$.
  Combining these inequalities we get
  \begin{align*}
    \kappa_d &\leq \frac{1}{1 - \eps} \kappa_m \\
             &\leq (1 + 2 \eps) (1 + \gamma) \kappa_m^* \\
             &\leq (1 + 2 \eps) (1 + \gamma) (1 + \eps) \kappa_d^* \\
             &\leq (1 + 4 \eps) (1 + \gamma) \kappa_d^*.
  \end{align*}
\end{proof}

By pushing the constant inside the $\Theta$-notation,
\cref{thm:kmeans:cost_jl} shows that we can achieve a $(1 + \eps)$
approximation\footnote{Here the approximation ratio is between any
  $k$-means algorithm running on the high dimensional original data
  and on the low dimensional projected data.} of $k$-means with
$m = \Theta(\eps^{-2} \log \abs{X})$. However, by more carefully
analysing which properties are needed, we can improve upon this for
the case where $k \ll \abs{X}$. \citet{Boutsidis:2014:RDRfkMC} showed
that projecting down to a target dimension of
$m = \Theta(\eps^{-2} k)$ suffices for a slightly worse $k$-means
approximation factor of $(2 + \eps)$. This result was expanded upon in
two ways by \citet{Cohen:2015:DRfkMCaLRA}, who showed that projecting
down to $m = \Theta(\eps^{-2} k)$ achieves a $(1 + \eps)$
approximation ratio, while projecting all the way down to
$m = \Theta(\eps^{-2} \log k)$ still suffices for a $(9 + \eps)$
approximation ratio. The $(1 + \eps)$ case has recently been further
improved upon by both \citet{Becchetti:2019:ODRfkMBSatJLL}, who have
shown that one can achieve the $(1 + \eps)$ approximation ratio for
$k$-means when projecting down to
$m = \Theta\bigl(\eps^{-6} (\log k + \log \log |X|) \log
\eps^{-1}\bigr)$, and by \citet{Makarychev:2019:PoJLTfkMakMC}, who
independently have proven an even better bound of
$m = \Theta\bigl(\eps^{-2} \log k/\eps \bigr)$, essentially giving a
``best of worlds'' result with respect to
\citep{Cohen:2015:DRfkMCaLRA}.

For an overview of the history of $k$-means clustering, we refer the
interested reader to \citep{Bock:2008:OaEotkMAiCA}.

\subsection{Streaming}
\label{sec:intro:streaming}


The field of streaming algorithms is characterised by problems where
we receive a sequence (or stream) of items and are queried on the
items received so far. The main constraint is usually that we only
have limited access to the sequence, e.g.\@ that we are only allowed
one pass over it, and that we have very limited space, e.g.\@
polylogarithmic in the length of the stream. To make up for these
constraints we are allowed to give approximate answers to the
queries. The subclass of streaming problems we will look at here are
those where we are only allowed a single pass over the sequence and
the items are updates to a vector and a query is some statistic on
that vector, e.g.\@ the $\ell_2$ norm of the vector. More formally,
and to introduce the notation, let $d \in \Npos$ be the number of
different items and let $T \in \Npos$ be the length of the stream of
updates $(i_j, v_j) \in [d] \times \R$ for $j \in [T]$, and define the
vector $x$ at time $t$ as $x^{(t)} \eqdef \sum_{j=1}^{t} v_j
e_{i_j}$. A query $q$ at time $t$ is then a function of $x^{(t)}$, and
we will omit the $^{(t)}$ superscript when referring to the current
time.

There are a few common variations on this model with respect to the
updates. In the \emph{cash register} model or \emph{insertion only}
model $x$ is only incremented by bounded integers, i.e.\@
$v_j \in [M]$, for some $M \in \Npos$. In the \emph{turnstile} model, $x$
can only be incremented or decremented by bounded integers, i.e.\@
$v_j \in \{-M, \dotsc, M\}$ for some $M \in \Npos$, and the \emph{strict
  turnstile} model is as the turnstile model with the additional
constraint that the entries of $x$ are always non-negative,
i.e.\@ $x^{(t)}_i \geq 0$, for all $t \in [T]$ and $i \in [d]$.

As mentioned above, we are usually space constrained so that we cannot
explicitely store $x$ and the key idea to overcome this limitation is
to store a linear \emph{sketch} of $x$, that is storing
$y \eqdef f(x)$, where $f \colon \R^d \to \R^m$ is a linear function
and $m \ll d$, and then answering queries by applying some function on
$y$ rather than $x$. Note that since $f$ is linear, we can apply it to
each update individually and compute $y$ as the sum of the sketched
updates. Furthermore, we can aggregate results from different streams
by adding the different sketches, allowing us to distribute the
computation of the streaming algorithm.

The relevant Johnson--Lindenstrauss lemma in this setting is
\cref{thm:distjl} as with a JLD we get linearity and are able to
sample a JLT before seeing any data at the cost of introducing some
failure probability.

Based on JLDs, the most natural streaming problem to tackle is second
frequency moment estimation in the turnstile model, i.e.\@
approximating $\|x\|_2^2$, which has found use in database query
optimisation~\citep{Alon:2002:TJaSJSiLS, Walton:1991:aTaPMoDSEiPJ,
  DeWitt:1992:PSHiPJ} and network data
analysis~\citep{Gilbert:2001:QQSaAoND, Cormode:2005:SSTtNDAQT} among
other areas. Simply letting $f$ be a sample from a JLD and returning
$\|f(x)\|_2^2$ on queries, gives a factor $(1 \pm \eps)$ approximation
with failure probability $\delta$ using
$\bigO(\eps^{-2} \log \frac{1}{\delta} + \abs{f})$ words\footnote{Here
  we assume that a word is large enough to hold a sufficient
  approximation of any real number we use and to hold a number from
  the stream, i.e.\@ if $w$ denotes the number of bits in a word then
  $w = \Omega(\log d + \log M)$.} of space, where $\abs{f}$ denotes
the words of space needed to store and apply $f$. However, the
approach taken by the streaming literature is to estimate $\|x\|_2^2$
with constant error probability using $\bigO(\eps^{-2} + \abs{f})$
words of space, and then sampling $\bigO(\log \frac{1}{\delta})$ JLTs
$f_1, \dotsc, f_{\bigO(\log 1 / \delta)} \colon \R^d \to
\R^{\bigO(\eps^{-2})}$ and responding to a query with
$\median_k \|f_k(x)\|_2^2$, which reduces the error probability to
$\delta$. This allows for simpler analyses as well more efficient
embeddings (in the case of Count Sketch) compared to using a single
bigger JLT, but it comes at the cost of not embedding into $\ell_2$,
which is needed for some applications outside of streaming. With this
setup the task lies in constructing space efficient JLTs and a seminal
work here is the AMS Sketch a.k.a.\@ AGMS Sketch a.k.a.\@ Tug-of-War
Sketch~\citep{Alon:1999:tSCoAtFM, Alon:2002:TJaSJSiLS}, whose JLTs can
be defined as $f_i \eqdef m^{-1/2}Ax$, where
$A \in \{-1, 1\}^{m \times d}$ is a random matrix. The key idea is
that each row $r$ of $A$ can be backed by a hash function
$\sigma_r \colon [d] \to \{-1, 1\}$ that need only be 4-wise
independent, meaning that for any set of 4 distinct keys
$\{k_1, \dotsc k_4\} \subset [d]$ and 4 (not necessarily distinct)
values $v_1, \dotsc v_4 \in \{-1, 1\}$, the probability that the keys
hash to those values is
$\Pr_{\sigma_r}[\bigwedge_i \sigma_r(k_i) = v_i] = \abs[\big]{\{-1,
  1\}}^{-4}$. This can for instance\footnote{See
  e.g.~\citep{Thorup:2012:TB5IHwAtLPaSME} for other families of
  $k$-wise independent hash functions.} be attained by implementing
$\sigma_r$ as 3rd degree polynomial modulus a sufficiently large prime
with random coefficients~\citep{Wegman:1981:NHFaTUiAaSE}, and so such
a JLT need only use $\bigO(\eps^{-2})$ words of space. Embedding a
scaled standard unit vector with such a JLT takes $\bigO(\eps^{-2})$
time leading to an overall update time of the AMS Sketch of
$\bigO(\eps^{-2} \log \frac{1}{\delta})$.

A later improvement of the AMS Sketch is the so-called Fast-AGMS
Sketch~\citep{Cormode:2005:SSTtNDAQT} a.k.a.\@ Count
Sketch~\citep{Charikar:2004:FFIiDS, Thorup:2012:TB5IHwAtLPaSME}, which
sparsifies the JLTs such that each column in their matrix
representations only has one non-zero entry. Each JLT can be
represented by a pairwise independent hash function
$h \colon [d] \to [\bigO(\eps^{-2})]$ to choose the position of each
nonzero entry and a 4-wise independent hash function
$\sigma \colon [d] \to \{-1, 1\}$ to choose random signs as
before. This reduces the standard unit vector embedding time to
$\bigO(1)$ and so the overall update time becomes
$\bigO(\log \frac{1}{\delta})$ for Count Sketch. It should be noted
that the JLD inside Count Sketch is also known as Feature Hashing,
which we will return to in \cref{sec:intro:sparse-jl}.

Despite not embedding into $\ell_2$, due to the use of the non-linear
median, AMS Sketch and Count Sketch approximately preserve dot
products similarly to \cref{thm:jl:dotprod}~\citep[][Theorem 2.1 and
Theorem 3.5]{Cormode:2005:SSTtNDAQT}. This allows us to query for the
(approximate) frequency of any particular item as
\begin{equation*}
  \median_k \dotprod{f_k(x)}{f_k(e_i)} = \dotprod{x}{e_i} \pm \eps\|x\|_2\|e_i\|_2 = x_i \pm \eps\|x\|_2
\end{equation*}
with probability at least $1 - \delta$.

This can be extended to finding frequent items in an insertion only
stream~\citep{Charikar:2004:FFIiDS}. The idea is to use a slightly
larger\footnote{Rather than each JLT having a target dimension of
  $\bigO(\eps^{-2})$, the analysis needs the target dimension to be
  $\bigO\Bigl( \frac{\|\tail_k(x)\|_2^2}{(\eps x_{i_k})^2} \Bigr)$,
  where $\tail_k(x)$ denotes $x$ with its $k$ largest entries zeroed
  out.} Count Sketch instance to maintain a heap of the $k$
approximately most frequent items of the stream so far. That is, if we
let $i_k$ denote the $k$th most frequent item (i.e.\@
$\bigl|\{ j \mid x_j \geq x_{i_k} \}\bigr| = k$), then with
probability $1 - \delta$ we have $x_j > (1 - \eps) x_{i_k}$ for every
item $j$ in our heap.

For more on streaming algorithms, we refer the reader to
\citep{Muthukrishnan:2005:DSAaA} and \citep{Nelson:2011:SaSHDV}, which
also relates streaming to Johnson--Lindenstrauss.

\section{The Tapestry of Johnson--Lindenstrauss Transforms}
\label{sec:intro:evolution-jl}

\epigraph{Isti mirant stella}{---\textup{Scene 32}, The Bayeux Tapestry~\citep{Bayeux:1070:tBT}}

\noindent
As mentioned in \cref{sec:intro:jl_lemma}, the original JLD from
\citep{Johnson:1984:EoLMiaHS} is a distribution over functions
$f \colon \R^{d} \to \R^{m}$, where\footnote{We will usually omit the
  normalisation or scaling factor (the $(d/m)^{1/2}$ for this JLD)
  when discussing JLDs as they are textually noisy, not that
  interesting, and independent of randomness and input data.}
$f(x) = (d/m)^{1/2}Ax$ and $A$ is a random $m \times d$ matrix whose
rows form an orthonormal basis of some $m$-dimensional subspace of
$\R^d$, i.e.\@ the rows are unit vectors and pairwise
orthogonal. While \citet{Johnson:1984:EoLMiaHS} showed that
$m = \Theta(\eps^{-2} \log |X|)$ suffices to prove \cref{thm:jl}, they
did not give any bounds on the constant in the big-$\bigO$
expression. This was remedied in \citep{Frankl:1988:tJLLatSosG}, which
proved that
$m = \bigl\lceil 9(\eps^2 - 2\eps^3/3)^{-1} \ln |X| \bigr\rceil + 1$
suffices for the same JLD if $m < \sqrt{|X|}$. This bound was further
improved in \citep{Frankl:1990:SGAotBD} by removing the
$m < \sqrt{|X|}$ restriction and lowering the bound to
$m = \bigl\lceil 8(\eps^2 - 2\eps^3/3)^{-1} \ln |X| \bigr\rceil$.

The next thread of JL research worked on simplifying the JLD
constructions as
\citeauthor{Indyk:1998:ANNTRtCoD}~\citep{Har-Peled:2012:ANNTRtCoD}
showed that sampling each entry in the matrix i.i.d.\@ from a properly
scaled Gaussian distribution is a JLD. The rows of such a matrix do
not form a basis as they are with high probability not orthogonal;
however, the literature still refer to this and most other JLDs as
random projections. Shortly thereafter \citet{Arriaga:2006:aAToLRCaRP}
constructed a JLD by sampling i.i.d.\@ from a Rademacher\footnote{The
  Rademacher distribution is the uniform distribution on $\{-1, 1\}$.}
distribution, and \citet{Achlioptas:2003:DFRPJLwBC} sparsified the
Rademacher construction such that the entries $a_{ij}$ are sampled
i.i.d.\@ with $\Pr[a_{ij} = 0] = 2/3$ and
$\Pr[a_{ij} = -1] = \Pr[a_{ij} = 1] = 1/6$. We will refer to such
sparse i.i.d.\@ Rademacher constructions as Achlioptas
constructions. The Gaussian and Rademacher results have later been
generalised~\citep{Matousek:2008:oVotJLL, Indyk:2007:NNPE,
  Klartag:2005:EPaRP} to show that a JLD can be constructed by
sampling each entry in a $m \times d$ matrix i.i.d.\@ from any
distribution with mean 0, variance 1, and a subgaussian
tail\footnote{A real random variable $X$ with mean 0 has a subgaussian
  tail if there exists constants $\alpha, \beta > 0$ such that for all
  $\lambda > 0$,
  $\Pr\bigl[|X| > \lambda\bigr] \leq \beta e^{- \alpha
    \lambda^2}$.}. It should be noted that these developments have a
parallel in the streaming literature as the previously mentioned AMS
Sketch~\citep{Alon:1999:tSCoAtFM,Alon:2002:TJaSJSiLS} is identical to
the Rademacher construction~\citep{Arriaga:2006:aAToLRCaRP}, albeit
with constant error probability.

As for the target dimension for these constructions,
\citep{Har-Peled:2012:ANNTRtCoD} proved that the Gaussian construction
is a JLD if
$m \geq 8(\eps^2 - 2\eps^3/3)^{-1} \bigl(\ln |X| + \bigO(\log
m)\bigr)$, which roughly corresponds to an additional additive
$\bigO(\eps^{-2}\log \log |X|)$ term over the original
construction. This additive $\log\log$ term was shaved off by the
proof in \citep{Dasgupta:2002:aEPoaToJaL}, which concerns itself with
the original JLD construction but can easily\footnote{The main part of
  the proof in \citep{Dasgupta:2002:aEPoaToJaL} is showing that the
  $\ell_2$ norm of a vector of i.i.d.\@ Gaussians is concentrated
  around the expected value. A vector projected with the Gaussian
  construction is distributed as a vector of i.i.d.\@ Gaussians.} be
adapted to the Gaussian construction, and the proof in
\citep{Arriaga:2006:aAToLRCaRP}, which also give the same $\log\log$
free bound for the dense Rademacher
construction. \citet{Achlioptas:2003:DFRPJLwBC} showed that his
construction also achieves
$m = \bigl\lceil 8(\eps^2 - 2\eps^3/3)^{-1} \ln |X| \bigr\rceil$. The
constant of 8 has been improved for the Gaussian and dense Rademacher
constructions in the sense that
\citeauthor{Rojo:2010:ItJLL}~\citep{Rojo:2010:ItJLL,
  Nguyen:2009:DRMwAtJDDwaCR} have been able to replace the bound with
more intricate\footnote{For example, one of the bounds for the
  Rademacher construction is
  $m \geq \frac{2(d - 1) \alpha^2}{d \eps^2}$, where
  $\alpha \eqdef \frac{Q + \sqrt{Q^2 + 5.98}}{2}$,
  $Q \eqdef \Phi^{-1}(1 - 1/|X|^2)$, and $\Phi^{-1}$ is the quantile
  function of the standard Gaussian random variable.} expressions,
which yield a 10 to 40~\% improvement for many sets of
parameters. However, in the distributional setting it has been shown
in \citep{Burr:2018:OBfJLT} that
$m \geq 4 \eps^{-2} \ln \frac{1}{\delta} (1 - o(1))$ is necessary for
any JLD to satisfy \cref{thm:distjl}, which corresponds to a constant
of $8$ if we prove \cref{thm:jl} the usual way by setting
$\delta = n^{-2}$ and union bounding over all pairs of vectors.

There seems to have been some confusion in the literature regarding
the improvements in target dimension. The main pitfall was that some
papers~\citep[e.g.\@][]{Achlioptas:2003:DFRPJLwBC,
  Har-Peled:2012:ANNTRtCoD, Dasgupta:2002:aEPoaToJaL, Rojo:2010:ItJLL,
  Burr:2018:OBfJLT} were only referring to
\citep{Frankl:1988:tJLLatSosG} when referring to the target dimension
bound of the original construction. As such,
\citep{Achlioptas:2003:DFRPJLwBC,Har-Peled:2012:ANNTRtCoD} mistakenly
claim to improve the constant for the target dimension with their
constructions.  Furthermore, \citep{Achlioptas:2001:DFRP} is sometimes
\citep[e.g.\@~in][]{Ailon:2009:tFJLTaANN, Matousek:2008:oVotJLL, Schmidt:2018:SRPLGPDRfDL} the only
work credited for the Rademacher construction, despite it being
developed independently and published 2 years prior in
\citep{Arriaga:1999:aAToLRCaRP}.
%

All the constructions that have been mentioned so far in this section,
embed a vector by performing a relatively dense and unstructured
matrix-vector multiplication, which takes
$\Theta(m \|x\|_0) = \bigO(m d)$ time\footnote{$\|x\|_0$ is the number
  of nonzero entries in the vector $x$.} to compute. This sparked two
distinct but intertwined strands of research seeking to reduce the
embedding time, namely the sparsity-based JLDs which dealt with the
density of the embedding matrix and the fast Fourier transform-based
which introduced more structure to the matrix.

\subsection{Sparse Johnson--Lindenstrauss Transforms}
\label{sec:intro:sparse-jl}

The simple fact underlying the following string of results is that if
$A$ has $s$ nonzero entries per column, then $f(x) = Ax$ can be
computed in $\Theta(s\|x\|_0)$ time. The first result here is the
Achlioptas construction~\citep{Achlioptas:2003:DFRPJLwBC} mentioned
above, whose column sparsity $s$ is $m/3$ in expectancy, which leads
to an embedding time that is a third of the full Rademacher
construction\footnote{Here we ignore any overhead that switching to a
  sparse matrix representation would introduce.}. However, the first
superconstant improvement is due to \citet{Dasgupta:2010:aSJLT}, who
based on heuristic approaches~\citep{Weinberger:2009:FHfLSML,
  Shi:2009:HKfSD, Langford:2007:VWCR, Ganchev:2008:SSMbRFM}
constructed a JLD with
$s = \bigO(\eps^{-1} \log \frac{1}{\delta} \log^2
\frac{m}{\delta})$. Their construction, which we will refer to as the
DKS construction, works by sampling $s$ hash functions
$h_1, \dotsc, h_s \colon [d] \to \{-1, 1\} \times [m]$ independently,
such that each source entry $x_i$ will be hashed to $s$ random signs
$\sigma_{i,1}, \dotsc, \sigma_{i,s}$ and $s$ target coordinates
$j_{i,1}, \dotsc, j_{i, s}$ (with replacement). The embedding can then
be defined as
$f(x) \eqdef \sum_{i} \sum_{k} e_{j_{i,k}} \sigma_{i,k} x_i$, which is
to say that every source coordinate is hashed to $s$ output
coordinates and randomly added to or subtracted from those output
coordinates. The sparsity analysis was later tightened to show that
$s = \bigO(\eps^{-1} \log\frac{1}{\delta} \log\frac{m}{\delta})$
suffices~\citep{Kane:2010:aDSJLT,Kane:2014:SJLT} and even
$s = \bigO\Bigl( \eps^{-1} \bigl(\frac{\log\frac{1}{\delta}
  \log\log\log\frac{1}{\delta}}{\log\log\frac{1}{\delta}} \bigr)^2
\Bigr)$ suffices for the DKS construction assuming
$\eps < \log^{-2}\frac{1}{\delta}$ \citep{Braverman:2010:RCREGatSJLT},
while \citep{Kane:2014:SJLT} showed that
$s = \Omega(\eps^{-1} \log^2\frac{1}{\delta} / \log^2\frac{1}{\eps})$
is neccessary for the DKS construction.

\citet{Kane:2014:SJLT} present two constructions that circumvent the
DKS lower bound by ensuring that the hash functions do not collide
within a column, i.e.\@ $j_{i,a} \neq j_{i, b}$ for all $i$, $a$, and
$b$. The first construction, which we will refer to as the graph
construction, simply samples the $s$ coordinates without
replacement. The second construction, which we will refer to as the
block construction, partitions the output vector into $s$ consecutive
blocks of length $m/s$ and samples one output coordinate per
block. Note that the block construction is the same as Count Sketch
from the streaming
literature~\citep{Cormode:2005:SSTtNDAQT,Charikar:2004:FFIiDS}, though
the hash functions differ and the output is interpreted
differently. \citet{Kane:2014:SJLT} prove that
$s = \Theta(\eps^{-1} \log \frac{1}{\delta})$ is both neccessary and
sufficient in order for their two constructions to satisfy
\cref{thm:distjl}. Note that while Count Sketch is even sparser than
the lower bound for the block construction, it does not contradict it
as Count sketch does not embed into $\ell_2^{m}$ as it computes the
median, which is nonlinear. As far as general sparsity lower bounds
go, \citep{Dasgupta:2010:aSJLT} shows that an average column sparsity
of
$s_{\mathrm{avg}} = \Omega\bigl(\min \{ \eps^{-2}, \eps^{-1}
\sqrt{\log_m d}\}\bigr)$ is neccessary for a sparse JLD, while
\citet{Nelson:2013:SLBfDRM} improves upon this by showing that there
exists a set of points $X \in \R^d$ such that any JLT for that set
must have column sparsity
$s = \Omega(\eps^{-1} \log |X| / \log \frac{1}{\eps})$ in order to
satisfy \cref{thm:jl}. And so it seems that we have almost reached the
limit of the sparse JL approach, but why should theory be in the way
of a good result? Let us massage the definitions so as to get around
these lower bounds.

The hard instances used to prove the lower
bounds~\citep{Nelson:2013:SLBfDRM,Kane:2014:SJLT} consist of very
sparse vectors, e.g.\@ $x = (1/\sqrt{2}, 1/\sqrt{2}, 0, \dotsc, 0)^\T$,
but the vectors we are interested in applying a JLT to might not be so
unpleasant, and so by restricting the input vectors to be sufficiently
``nice'', we can get meaningful result that perform better than what
the pessimistic lower bound would indicate. The formal formulation of
this niceness is bounding the $\ell_\infty / \ell_2$ ratio of the
vectors \cref{thm:jl,thm:distjl} need apply to. Let us denote this
norm ratio as $\nu \in [1/\sqrt{d}, 1]$, and revisit some of the
sparse JLDs. The Achlioptas
construction~\citep{Achlioptas:2003:DFRPJLwBC} can be generalised so
that the expected number of nonzero entries per column is $qm$ rather
than $\frac{1}{3}m$ for a parameter $q \in [0,
1]$. \citet{Ailon:2009:tFJLTaANN} show that if
$\nu = \bigO\bigl( \sqrt{\log \frac{1}{\delta}} / \sqrt{d} \bigr)$
then choosing $q = \Theta\bigl(\frac{\log^2 1/\delta}{d}\bigr)$ and
sampling the nonzero entries from a Gaussian distribution
suffices. This result is generalised in \citep{Matousek:2008:oVotJLL}
by proving that for all $\nu \in [1/\sqrt{d}, 1]$ choosing
$q = \Theta(\nu^2 \log \frac{d}{\eps\delta})$ and sampling the nonzero
entries from a Rademacher distribution is a JLD for the vectors
constrained by that $\nu$.

Be aware that sometimes~\citep[e.g.\@~in][]{Dasgupta:2010:aSJLT,
  Braverman:2010:RCREGatSJLT} this bound\footnote{Which seems to be
  the only thing in \citep{Matousek:2008:oVotJLL} related to a bound
  on $q$.} on $q$ is misinterpreted as a lower bound stating that
$qm = \tilde\Omega(\eps^{-2})$ is neccessary for the Achlioptas
construction when $\nu = 1$. However, \citet{Matousek:2008:oVotJLL}
only loosely argues that his bound is tight for $\nu \leq d^{-0.1}$,
and if it indeed was tight at $\nu = 1$, the factors hidden by the
$\tilde\Omega$ would lead to the contradiction that
$m \geq qm = \Omega(\eps^{-2} \log \frac{1}{\delta} \log \frac{d}{\eps
  \delta}) = \omega(m)$.

The heuristic~\citep{Weinberger:2009:FHfLSML, Shi:2009:HKfSD,
  Langford:2007:VWCR, Ganchev:2008:SSMbRFM} that
\citep{Dasgupta:2010:aSJLT} is based on is called Feature Hashing
a.k.a.\@ the hashing trick a.k.a.\@ the hashing kernel and is a sparse
JL construction with exactly 1 nonzero entry per
column\footnote{i.e.\@ the DKS, graph, or block construction with
  $s = 1$.}. The block construction can then be viewed as the
concatenation of $s = \Theta(\eps^{-1} \log \frac{1}{\delta})$ feature
hashing instances, and the DKS construction can be viewed as the sum
of $s = \bigO(\eps^{-1} \log\frac{1}{\delta} \log\frac{m}{\delta})$
Feature Hashing instances or alternatively as first duplicating each
entry of $x \in \R^d$ $s$ times before applying Feature Hashing to the
enlarged vector $x' \in \R^{sd}$: Let
$f_{\mathsf{dup}} \colon \R^d \to \R^{sd}$ be a function that
duplicates each entry in its input $s$ times, i.e.\@
$f_{\mathsf{dup}}(x)_{(i - 1)s + j} = x'_{(i - 1)s + j} \eqdef x_i$
for $i \in [d], j \in [s]$, then
$f_{\mathsf{DKS}} = f_{\mathsf{FH}} \circ f_{\mathsf{dup}}$.

This duplication is the key to the analysis in
\citep{Dasgupta:2010:aSJLT} as $f_{\mathsf{dup}}$ is isometric (up to
normalisation) and it ensures that the $\ell_\infty / \ell_2$ ratio of
$x'$ is small, i.e.\@ $\nu \leq 1 / \sqrt{s}$ from the point of view
of the Feature Hashing data structure ($f_{\mathsf{FH}}$). And so, any
lower bound on the sparsity of the DKS construction (e.g.\@ the one
given in \citep{Kane:2014:SJLT}) gives an upper bound on the values of
$\nu$ for which Feature Hashing is a JLD: If $u$ is a unit vector such
that a DKS instance with sparsity $\hat{s}$ fails to preserve $u$s norm
within $1\pm\eps$ with probability $\delta$, then it must be the case
that Feature Hashing fails to preserve the norm of
$f_{\mathsf{dup}}(u)$ within $1\pm\eps$ with probability $\delta$, and
therefore the $\ell_\infty / \ell_2$ ratio for which Feature Hashing
can handle all vectors is strictly less than $1 / \sqrt{\hat{s}}$.

Written more concisely the statement is
$s_{\mathsf{DKS}} = \Omega(a) \implies \nu_{\mathsf{FH}} = \bigO(1 /
\sqrt{a})$ and by contraposition\footnote{Contraposition is
  $(P \implies Q) \implies (\neg Q \implies \neg P)$ and it does not
  quite prove what was just claimed without some assumptions that
  $s_{\mathsf{DKS}}$, $\nu_{\mathsf{FH}}$, and $a$ do not behave too
  erratically.}
$\nu_{\mathsf{FH}} = \Omega(1 / \sqrt{a}) \implies s_{\mathsf{DKS}} =
\bigO(a)$, where $s_{\mathsf{DKS}}$ is the minimum column sparsity of
a DKS construction that is a JLD, $\nu_{\mathsf{FH}}$ is the maximum
$\ell_\infty / \ell_2$ constraint for which Feature Hashing is a JLD,
and $a$ is any positive expression. Furthermore, if we prove an upper
bound on $\nu_{\mathsf{FH}}$ using a hard instance that is identical
to an $x'$ that the DKS construction can generate after duplication,
we can replace the previous two implications with bi-implications.

\citep{Weinberger:2009:FHfLSML} claims to give a bound on
$\nu_{\mathsf{FH}}$, but it sadly contains an error in its proof of
this bound~\citep{Dasgupta:2010:aSJLT,
  Weinberger:2010:FHfLSML}. \citet{Dahlgaard:2017:PHFfSEaDR} improve
the $\nu_{\mathsf{FH}}$ lower bound to
$\nu_{\mathsf{FH}} = \Omega\Bigl( \sqrt{\frac{\eps \log(1 +
    \frac{4}{\eps})} {\log \frac{1}{\delta} \log \frac{m}{\delta}}}
\Bigr)$, and \citet{Freksen:2018:FUtHT} give an intricate but tight
bound for $\nu_{\mathsf{FH}}$ shown in \cref{thm:fh:main}, where the
hard instance used to prove the upper bound is identical to an $x'$
from the DKS construction.
\begin{theorem}[\citep{Freksen:2018:FUtHT}]
  \label{thm:fh:main}
  There exist constants $C \geq D > 0$ such
  that for every $\eps, \delta \in (0,1)$ and
  $m \in \Npos$ the following holds. If
  $\frac{C \lg \frac{1}{\delta}}{\eps^2} \leq m <
  \frac{2}{\eps^2 \delta}$ then
  \begin{equation*}
    \nu_{\mathsf{FH}}(m, \eps, \delta) = \Theta\biggl( \sqrt{\eps} \min\Bigl\{ \frac{\log\frac{\eps m}{\log \frac{1}{\delta}}}{\log\frac{1}{\delta}}, \sqrt{\frac{\log \frac{\eps^2 m}{\log \frac{1}{\delta}}}{\log\frac{1}{\delta}}} \Bigr\} \biggr) .
  \end{equation*}
  Otherwise, if $m \geq \frac{2}{\eps^2 \delta}$ then
  $\nu_{\mathsf{FH}}(m, \eps, \delta) = 1$. Moreover if
  $m < \frac{D\lg\frac{1}{\delta}}{\eps^2}$ then
  $\nu_{\mathsf{FH}}(m, \eps, \delta)=0$.

  Furthermore, if an $x \in \{0, 1\}^d$ satisfies
  $\nu_{\mathsf{FH}} < \|x\|_2^{-1} < 1$ then
  \begin{equation*}
    \Pr_{f \sim \mathsf{FH}}\Bigl[\bigl| \|f(x)\|_2^2 - \|x\|_2^2 \bigr| > \eps\|x\|_2^2 \Bigr] > \delta.
  \end{equation*}
\end{theorem}
This bound gives a tight tradeoff between target dimension $m$,
distortion $\eps$, error probability $\delta$, and
$\ell_\infty / \ell_2$ constraint $\nu$ for Feature Hashing, while
showing how to construct hard instances for Feature Hashing: Vectors
with the shape $x = (1, \dotsc, 1, 0, \dotsc, 0)^\T$ are hard
instances if they contain few $1$s, meaning that Feature Hashing
cannot preserve their norms within $1 \pm \eps$ with probability
$\delta$. \Cref{thm:fh:main} is used in \cref{thm:dks:supersparse} to
provide a tight tradeoff between $m$, $\eps$, $\delta$, $\nu$, and
column sparsity $s$ for the DKS construction.
\begin{corollary}
  \label{thm:dks:supersparse}
  Let $\nu_{\mathsf{DKS}} \in [1/\sqrt{d}, 1]$ denote the largest
  $\ell_{\infty} / \ell_2$ ratio required, $\nu_{\mathsf{FH}}$ denote
  the $\ell_{\infty} / \ell_2$ constraint for Feature Hashing as
  defined in \cref{thm:fh:main}, and $s_{\mathsf{DKS}} \in [m]$ as the
  minimum column sparsity such that the DKS construction with that
  sparsity is a JLD for the subset of vectors $x \in \R^d$ that
  satisfy $\|x\|_\infty / \|x\|_2 \leq \nu_{\mathsf{DKS}}$. Then
  \begin{equation*}
    s_{\mathsf{DKS}} = \Theta\Bigl( \frac{\nu_{\mathsf{DKS}}^2}{\nu_{\mathsf{FH}}^2} \Bigr).
  \end{equation*}
\end{corollary}
The proof of this corollary is deferred to
\cref{sec:app:dks-supersparse}.

\citet{Jagadeesan:2019:USJLFfFH} generalised the result from
\citep{Freksen:2018:FUtHT} to give a lower bound\footnote{Here a lower
  bound refers to a lower bound on $\nu$ as a function of $m$, $\eps$,
  $\delta$, and $s$.} on the $m$, $\eps$, $\delta$, $\nu$, and $s$
tradeoff for any sparse Rademacher construction with a chosen column
sparsity, e.g.\@ the block and graph constructions, and gives a matching
upper bound for the graph construction.

\subsection{Structured Johnson--Lindenstrauss Transforms}
\label{sec:intro:fft-jl}

As we move away from the sparse JLDs we will slightly change our idea
of what an efficient JLD is. In the previous section the JLDs were
especially fast when the vectors were sparse, as the running time
scaled with $\norm{x}{0}$, whereas we in this section will optimise
for dense input vectors such that an embedding time of
$\bigO(d \log d)$ is a satisfying result.

The chronologically first asymptotic improvement over the original JLD
construction is due to \citet{Ailon:2009:tFJLTaANN} who introduced the
so-called Fast Johnson--Lindenstrauss Transform (FJLT). As mentioned
in the previous section, \citep{Ailon:2009:tFJLTaANN} showed that we
can use a very sparse (and therefore very fast) embedding matrix as
long as the vectors have a low $\ell_\infty / \ell_2$ ratio, and
furthermore that applying a randomised Walsh--Hadamard transform to a
vector results in a low $\ell_\infty / \ell_2$ ratio with high
probability. And so, the FJLT is defined as $f(x) \eqdef PHDx$, where
$P \in \R^{m \times d}$ is a sparse Achlioptas matrix with Gaussian
entries and $q = \Theta\bigl(\frac{\log^21 / \delta}{d}\bigr)$,
$H \in \{-1, 1\}^{d \times d}$ is a Walsh--Hadamard
matrix\footnote{One definition of the Walsh--Hadamard matrices is that
  the entries are $H_{ij} = (-1)^{\dotprod{i-1}{j-1}}$ for all
  $i, j \in [d]$, where $\dotprod{a}{b}$ denote the dot product of the
  ($\lg d$)-bit vectors corresponding to the binary representation of
  the numbers $a, b \in \{0, \dotsc, d-1\}$, and $d$ is a power of
  two. To illustrate its recursive nature, a large Walsh--Hadamard
  matrix can be described as a Kronecker product of smaller
  Walsh--Hadamard matrices, i.e. if $d > 2$ and $H^{(n)}$ refers to a
  $n \times n$ Walsh--Hadamard matrix, then
  $H^{(d)} = H^{(2)} \otimes H^{(d/2)} =
  \begin{pmatrix}
    H^{(d/2)} & H^{(d/2)} \\
    H^{(d/2)} & -H^{(d/2)}
  \end{pmatrix}$.}, and $D \in \{-1, 0, 1\}^{d \times d}$ is a random
diagonal matrix with i.i.d.\@ Rademachers on the diagonal. As the
Walsh--Hadamard transform can be computed using a simple recursive
formula, the expected embedding time becomes
$\bigO(d \log d + m \log^2 \frac{1}{\delta})$. And as mentioned,
\citep{Matousek:2008:oVotJLL} showed that we can sample from a
Rademacher rather than a Gaussian distribution when constructing the
matrix $P$. The embedding time improvement of FJLT over previous
constructions depends on the relationship between $m$ and $d$. If
$m = \Theta(\eps^{-2} \log \frac{1}{\delta})$ and
$m = \bigO(\eps^{-4/3} d^{1/3})$, FJLT's embedding time becomes
bounded by the Walsh--Hadamard transform at $\bigO(d \log d)$, but at
$m = \Theta(d^{1/2})$ FJLT is only barely faster than the original
construction.

\citet{Ailon:2009:FDRuRSoDBCHC} improved the running time of the FJLT
construction to $\bigO(d \log m)$ for $m = \bigO(d^{1/2 - \gamma})$
for any fixed $\gamma > 0$. The increased applicable range of $m$ was
achieved by applying multiple randomised Walsh--Hadamard
transformations, i.e.\@ replacing $HD$ with $\prod_i HD^{(i)}$, where
the $D^{(i)}$s are a constant number of independent diagonal
Rademacher matrices, as well as by replacing $P$ with $BD$ where $D$
is yet another diagonal matrix with Rademacher entries and $B$ is
consecutive blocks of specific partial Walsh--Hadamard matrices (based
on so-called binary dual BCH codes~\citep[see
e.g.\@][]{MacWilliams:1977:tToECC}). The reduction in running time comes
from altering the transform slightly by partitioning the input into
consecutive blocks of length $\text{poly}(m)$ and applying the
randomised Walsh--Hadamard transforms to each of them
independently. We will refer to this variant of FJLT as the BCHJL
construction.

The next pattern of results has roots in compressed sensing and
approaches the problem from another angle: Rather than being fast only
when $m \ll d$, they achieve $\bigO(d \log d)$ embedding time even
when $m$ is close to $d$, at the cost of $m$ being suboptimal. Before
describing these constructions, let us set the scene by briefly
introducing some concepts from compressed sensing.

Roughly speaking, compressed sensing concerns itself with recovering a
sparse signal via a small number of linear measurements and a key
concept here is the Restricted Isometry
Property~\citep{Candes:2005:DbLP, Candes:2006:NOSRfRPUES,
  Candes:2006:RUPESRfHIFI, Donoho:2006:fMLUSoEtMl1NNSAtSNS}.
\begin{definition}[Restricted Isometry Property]
  \label{def:rip}
  Let $d, m, k \in \Npos$ with $m, k < d$ and $\eps \in (0, 1)$. A linear
  function $f \colon \R^d \to \R^m$ is said to have the Restricted
  Isometry Property of order $k$ and level $\eps$ (which we will
  denote as $(k, \eps)$-RIP) if for all $x \in \R^d$ with
  $\|x\|_0 \leq k$,
  \begin{equation}
    \label{eq:rip}
    \bigl| \|f(x)\|_2^2 - \|x\|_2^2 \bigr| \leq \eps \|x\|_2^2.
  \end{equation}
\end{definition}

In the compressed sensing literature it has been
shown~\citep{Candes:2006:NOSRfRPUES,Rudelson:2008:oSRfFaGm} that the
subsampled Hadamard transform (SHT) defined as $f(x) \eqdef SHx$, has
the $(k, \eps)$-RIP with high probability for
$m = \Omega\bigl(\eps^{-2} k \log^4d \bigr)$ while allowing a vector
to be embedded in $\bigO(d \log d)$ time. Here
$H \in \{-1, 1\}^{d \times d}$ is the Walsh--Hadamard matrix and
$S \in \{0, 1\}^{m \times d}$ samples $m$ entries of $Hx$ with
replacement, i.e.\@ each row in $S$ has one non-zero entry per row,
which is chosen uniformly and independently, i.e.\@ $S$ is a uniformly
random feature selection matrix. Inspired by this transform and the
FJLT mentioned previously, \citet{Ailon:2013:aAOUFJLT} were able to
show that the subsampled randomised Hadamard transform (SRHT) defined
as $f(x) \eqdef SHDx$, is a JLT if
$m = \Theta(\eps^{-4} \log |X| \log^4 d)$. Once again $D$ denotes a
random diagonal matrix with Rademacher entries, and $S$ and $H$ is as
in the SHT. Some related results include \citet{Do:2009:FaEDRuSRM} who
before \citep{Ailon:2013:aAOUFJLT} were able to get a bound of
$m = \Theta(\eps^{-2} \log^3 |X|)$ in the large set case where
$|X| \geq d$, \citep{Tropp:2011:IAotSRHT} which showed how the SRHT
construction approximately preserves the norms of a subspace of
vectors, and \citep{Lei:2020:ISRHTfLSVM} which modified the sampling
matrix $S$ to improve precision when used as a preprocessor for
support vector machines (SVMs) by sacrificing input data independence.

This target dimension bound of \citep{Ailon:2013:aAOUFJLT} was later
tightened by \citet{Krahmer:2011:NaIJLEvtRIP}, who showed that
$m = \Theta(\eps^{-2} \log |X| \log^4 d)$ suffices for the SRHT. This
was a corollary of a more general result, namely that if
$\sigma \colon \R^d \to \R^d$ applies random signs equivalently to the
$D$ matrices mentioned previously and $f \colon \R^d \to \R^m$ has the
$\bigl(\Omega(\log |X|), \eps/4\bigr)$-RIP then $f \circ \sigma$ is a
JLT with high probability. An earlier result by
\citet{Baraniuk:2008:aSPotRIPfRM} showed that a transform sampled from
a JLD has the $\bigl(\bigO(\eps^2 m / \log d), \eps\bigr)$-RIP with
high probability. And so, as one might have expected from their
appearance, the Johnson--Lindenstrauss Lemma and the Restricted
Isometry Property are indeed cut from the same cloth.

Another transform from the compressed sensing literature uses
so-called Toeplitz or partial circulant
matrices~\citep{Bajwa:2007:TSCSM, Rauhut:2009:CaTMiCS,
  Romberg:2009:CSbRC, Haupt:2010:TCSMwAtSCE, Rauhut:2012:RIfPRCM,
  Bajwa:2012:GoRTBSMBaIfSSP, Dirksen:2019:obCSwPGCM}, which can be
defined in the following way. For $m, d \in \Npos$ we say that
$T \in \R^{m \times d}$ is a real Toeplitz matrix if there exists
$t_{-(m - 1)}, t_{-(m - 2)} \dotsc, t_{d - 1} \in \R$ such that
$T_{ij} = t_{j-i}$. This has the effect that the entries on any one
diagonal are the same (see \cref{fig:intro:toeplitz}) and computing
the matrix-vector product corresponds to computing the convolution
with a vector of the $t$s. Partial circulant matrices are special
cases of Toeplitz matrices where the diagonals ``wrap around'' the
ends of the matrix, i.e.\@ $t_{-i} = t_{d-i}$ for all $i \in [m-1]$.

\begin{figure}
  \begin{equation*}
    \begin{pmatrix}
      t_0 & t_1 & t_2 & \cdots & t_{d-1} \\
      t_{-1} & t_0 & t_1 & \cdots & t_{d-2} \\
      t_{-2} & t_{-1} & t_0 & \cdots & t_{d-3} \\
      \vdots & \vdots & \vdots & \ddots & \vdots \\
      t_{-(m-1)} & t_{-(m-2)} & t_{-(m-3)} & \cdots & t_{d - m}
    \end{pmatrix}
  \end{equation*}
  \caption{The structure of a Toeplitz matrix.}
  \label{fig:intro:toeplitz}
\end{figure}

As a JLT, the Toeplitz construction is $f(x) \eqdef TDx$, where
$T \in \{-1, 1\}^{m \times d}$ is a Toeplitz matrix with i.i.d.\@
Rademacher entries and $D \in \{-1, 0, 1\}^{d \times d}$ is a diagonal
matrix with Rademacher entries as usual. Note that the convolution of
two vectors corresponds to the entrywise product in Fourier space, and
we can therefore employ fast Fourier transform (FFT) to embed a vector
with the Toeplitz construction in time $\bigO(d \log d)$. This time
can even be reduced to $\bigO(d \log m)$ as we realise that by
partitioning $T$ into $\frac{d}{m}$ consecutive blocks of size
$m \times m$, each block is also a Toeplitz matrix, and by applying
each individually the embedding time becomes
$\bigO(\frac{d}{m} m \log m)$.

Combining the result from \citep{Krahmer:2011:NaIJLEvtRIP} with RIP
bounds for Toeplitz matrices~\citep{Rauhut:2012:RIfPRCM} gives that
$m = \Theta\bigl( \eps^{-1} \log^{3/2} |X| \log^{3/2} d + \eps^{-2}
\log |X| \log^4 d \bigr)$ is sufficient for the Toeplitz construction
to be a JLT with high probability. However, the Toeplitz construction
has also been studied directly as a JLD without going via its RIP
bounds. \citet{Hinrichs:2011:JLLfCM} showed that
$m = \Theta(\eps^{-2} \log^3 \frac{1}{\delta})$ is sufficient for the
Toeplitz construction, and this bound was improved shortly thereafter
in \citep{Vybiral:2011:aVotJLLfCM} to
$m = \Theta(\eps^{-2} \log^2 \frac{1}{\delta})$. The question then is
if we can tighten the analysis to shave off the last $\log$ factor and
get the elusive result of a JLD with optimal target dimension and
$\bigO(d \log d)$ embedding time even when $m$ is close to $d$. Sadly,
this is not the case as \citet{Freksen:2020:oUTaCMfJLT} showed that
there exists vectors\footnote{Curiously, the hard instances for the
  Toeplitz construction are very similar to the hard instances for
  Feature Hashing used in \citep{Freksen:2018:FUtHT}.}  that
necessitates $m = \Omega(\eps^{-2} \log^2 \frac{1}{\delta})$ for the
Toeplitz construction.

Just as JLTs are used as preprocessors to speed up algorithms that
solve the problems we actually care about, we can also use a JLT to
speed up other JLTs in what one could refer to as compound JLTs. More
explicitely if $f_1 \colon \R^d \to \R^{d'}$ and
$f_2 \colon \R^{d'} \to \R^m$ with $m \ll d' \ll d$ are two JLTs and
computing $f_1(x)$ is fast, we could hope that computing
$(f_2 \circ f_1)(x)$ is fast as well as $f_2$ only need to handle $d'$
dimensional vectors and hope that $(f_2 \circ f_1)$ preserves the norm
sufficiently well since both $f_1$ and $f_2$ approximately preserve
norms individually. As presented here, the obvious candidate for $f_1$
is one of the RIP-based JLDs, which was succesfully applied in
\citep{Bamberger:2017:OFJLEfLDS}. In their construction, which we will
refer to as GRHD\footnote{Due to the choice of matrix names in
  \citep{Bamberger:2017:OFJLEfLDS}.}, $f_1$ is the SRHT and $f_2$ is
the dense Rademacher construction (i.e.\@
$f(x) \eqdef A_{\mathsf{Rad}}SHDx$), and it can embed a vector in time
$\bigO(d \log m)$ for $m = \bigO(d^{1/2-\gamma})$ for any fixed
$\gamma > 0$. This is a similar result to the construction of
\citet{Ailon:2009:FDRuRSoDBCHC}, but unlike that construction, GRHD
handles the remaining range of $m$ more gracefully as for any
$r \in [1/2, 1]$ and $m = \bigO(d^r)$, the embedding time for GRHD
becomes $\bigO(d^{2r} \log^4 d)$. However the main selling point of
the GRHD construction is that it allows the simultaneous embedding of
sufficiently large sets of points $X$ to be computed in total time
$\bigO(|X| d \log m)$, even when $m = \Theta(d^{1-\gamma})$ for any
fixed $\gamma > 0$, by utilising fast matrix-matrix multiplication
techniques~\citep{Lotti:1983:otACoRMM}.

Another compound JLD is based on the so-called lean Walsh transforms
(LWT)~\citep{Liberty:2011:DFRPaLWT}, which are defined based on
so-called seed matrices. For $r, c \in \Npos$ we say that
$A_1 \in \CC^{r \times c}$ is a seed matrix if $r < c$, its columns
are of unit length, and its rows are pairwise orthogonal and have the
same $\ell_2$ norm. As such, partial Walsh--Hadamard matrices and
partial Fourier matrices are seed matrices (up to normalisation);
however, for simplicity's sake we will keep it real by focusing on
partial Walsh--Hadamard matrices. We can then define a LWT of order
$l \in \Npos$ based on this seed as
$A_l \eqdef A_1^{\otimes l} = A_1 \otimes \dotsb \otimes A_1$, where
$\otimes$ denotes the Kronecker product, which we will quickly
define. Let $A$ be a $m \times n$ matrix and $B$ be a $p \times q$
matrix, then the Kronecker product $A \otimes B$ is the $mp \times nq$
block matrix defined as
\begin{equation*}
  A \otimes B \eqdef
      \begin{pmatrix}
      A_{11}B & \cdots & A_{1n}B \\
      \vdots  & \ddots & \vdots \\
      A_{m1}B & \cdots & A_{mn}B
    \end{pmatrix}.
\end{equation*}
Note that $A_l$ is a $r^l \times c^l$ matrix and that any
Walsh--Hadamard matrix can be written as $A_l$ for some $l$ and the
$2 \times 2$ Walsh--Hadamard matrix\footnote{Here we ignore the
  $r < c$ requirement of seed matrices.} as $A_1$. Furthermore, for a
constant sized seed the time complexity of applying $A_l$ to a vector
is $O(c^l)$ by using an algorithm similar to FFT. We can then define
the compound transform which we will refer to as LWTJL, as
$f(x) \eqdef GA_lDx$, where $D \in \{-1, 1\}^{d \times d}$ is a
diagonal matrix with Rademacher entries, $A_l \in \R^{r^l \times d}$
is a LWT, and $G \in \R^{m \times r^l}$ is a JLT, and $r$ and $c$ are
constants. One way to view LWTJL is as a variant of GRHD where the
subsampling occurs on the seed matrix rather than the final
Walsh--Hadamard matrix. If $G$ can be applied in $\bigO(r^l \log r^l)$
time, e.g.\@ if $G$ is the BCHJL
construction~\citep{Ailon:2009:FDRuRSoDBCHC} and
$m = \bigO(r^{l(1/2 - \gamma)})$, the total embedding time becomes
$\bigO(d)$, as $r^l = d^\alpha$ for some $\alpha < 1$. However, in
order to prove that LWTJL satisfies \cref{thm:distjl} the analysis of
\citep{Liberty:2011:DFRPaLWT} imposes a few requirements on $r$, $c$,
and the vectors we wish to embed, namely that
$\log r / \log c \geq 1 - 2\delta$ and
$\nu = \bigO(m^{-1/2} d^{-\delta})$, where $\nu$ is an upper bound on
the $\ell_\infty / \ell_2$ ratio as introduced at the end of
\cref{sec:intro:sparse-jl}. The bound on $\nu$ is somewhat tight as
shown in \cref{thm:lwt:nu_upper_bound}.
\begin{proposition}
  \label{thm:lwt:nu_upper_bound}
  For any seed matrix define $\mathsf{LWT}$ as the LWTJL distribution
  seeded with that matrix. Then for all $\delta \in (0, 1)$, there
  exists a vector $x \in \CC^d$ (or $x \in \R^d$, if the seed matrix
  is a real matrix) satisfying
  $\|x\|_\infty / \|x\|_2 = \Theta(\log^{-1/2} \frac{1}{\delta})$ such
  that
  \begin{equation}
    \label{eq:lwt:prob0}
    \Pr_{f \sim \mathsf{LWT}}[f(x) = \zerovector] > \delta.
  \end{equation}
\end{proposition}
The proof of \cref{thm:lwt:nu_upper_bound} can be found in
\cref{sec:app:lwt-nu-bound}, and it is based on constructing $x$ as a
few copies of a vector that is orthogonal to the rows of the seed
matrix.

The last JLD we will cover is based on so-called Kac random walks, and
despite \citet{Ailon:2009:tFJLTaANN} conjecturing that such a
construction could satisfy \cref{thm:jl}, it was not until
\citet{Jain:2020:KMJaLaMOFJLT} that a proof was finally at hand. As
with the lean Walsh transforms above, let us first define Kac random
walks before describing how they can be used to construct JLDs. A Kac
random walk is a Markov chain of linear transformations, where for
each step we choose two coordinates at random and perform a random
rotation on the plane spanned by these two coordinates, or more
formally:
\begin{definition}[Kac random walk~\citep{Kac:1956:FoKT}]
  For a given dimention $d \in \Npos$, let
  $K^{(0)} \eqdef I \in \{0, 1\}^{d \times d}$ be the identity matrix,
  and for each $t > 0$ sample $(i_t, j_t) \in \binom{[d]}{2}$ and
  $\theta_t \in [0, 2\pi)$ independently and uniformly at random. Then
  define the Kac random walk of length $t$ as
  $K^{(t)} \eqdef R^{(i_t, j_t, \theta_t)} K^{(t - 1)}$, where
  $R^{(i, j, \theta)} \in \R^{d \times d}$ is the rotation in the
  $(i, j)$ plane by $\theta$ and is given by
  \begin{alignat*}{4}
    R^{(i, j, \theta)} e_k &\eqdef e_k &&\forall k \notin \{i, j\}, \\
    R^{(i, j, \theta)} (a e_i + b e_j) &\eqdef (a \cos \theta - b \sin \theta)e_i + (a \sin \theta + b \cos \theta) e_j &&.
  \end{alignat*}
\end{definition}

The main JLD introduced in \citep{Jain:2020:KMJaLaMOFJLT}, which we
will refer to as KacJL, is a compound JLD where both $f_1$ and $f_2$
consists of a Kac random walk followed by subsampling, which can be
defined more formally in the following way. Let
$T_1 \eqdef \Theta(d \log d)$ be the length of the first Kac random
walk,
$d' \eqdef \min\{d, \Theta(\eps^{-2} \log |X| \log^2 \log |X| \log^3
d) \}$ be the intermediate dimension, $T_2 \eqdef \Theta(d' \log |X|)$
be the length of the second Kac random walk, and
$m \eqdef \Theta(\eps^{-2} \log |X|)$ be the target dimension, and
then define the JLT as
$f(x) = (f_2 \circ f_1)(x) \eqdef S^{(m, d')} K^{(T_2)} S^{(d', d)}
K^{(T_1)} x$, where $K^{(T_1)} \in \R^{d \times d}$ and
$K^{(T_2)} \in \R^{d' \times d'}$ are independent Kac random walks of
length $T_1$ and $T_2$, respectively, and
$S^{(d', d)} \in \{0, 1\}^{d' \times d}$ and
$S^{(m, d')} \in \{0, 1\}^{m \times d'}$ projects onto the first $d'$
and $m$ coordinates\footnote{The paper lets $d'$ and $m$ be random
  variables, but with the way the JLD is presented here a
  deterministic projection suffices, though it may affect constants
  hiding in the big-$\bigO$ expressions.}, respectively. Since
$K^{(T)}$ can be applied in time $\bigO(T)$, the KacJL construction is
JLD with embedding time
$\bigO(d \log d + \min \{d \log |X|, \eps^{-2} \log^2 |X| \log^2 \log
|X| \log^3 d \} )$ with asymptotically optimal target dimension, and
by only applying the first part ($f_1$), KacJL achieves an embedding
time of $\bigO(d \log d)$ but with a suboptimal target dimension of
$\bigO(\eps^{-2} \log |X| \log^2 \log |X| \log^3 d)$.

\citet{Jain:2020:KMJaLaMOFJLT} also proposes a version of their JLD
construction that avoids computing trigonometric
functions\footnote{Recall that $\sin (\pi/4) = 2^{-1/2}$ and that
  similar results holds for cosine and for the other angles.} by
choosing the angles $\theta_t$ uniformly at random from the set
$\{\pi/4, 3\pi/4, 5\pi/4, 7\pi/4\}$ or even the singleton set
$\{\pi/4\}$. This comes at the cost\footnote{Note that the various Kac
  walk lengths are only shown to be sufficient, and so tighter
  analysis might shorten them and perhaps remove the cost of using
  simpler angles.} of increasing $T_1$ by a factor of $\log \log d$
and $T_2$ by a factor of $\log d$, and for the singleton case
multiplying with random signs (as we have done with the $D$ matrices
in many of the previous constructions) and projecting down onto a
random subset of coordinates rather than the $d'$ or $m$ first.

\sectionbreak

This concludes the overview of Johnson--Lindenstrauss distributions
and transforms, though there are many aspects we did not cover such as
space usage, preprocessing time, randomness usage, and norms other
than $\ell_2$. However, a summary of the main aspects we did cover
(embedding times and target dimensions of the JLDs) can be found in
\cref{tab:intro:tapestry}.

\clearpage
\pagestyle{empty}
\global\pdfpageattr\expandafter{\the\pdfpageattr/Rotate 90}
\begin{sidewaystable}[p]
  \centerfloat
  \begin{tabular}{l l l l l}
    \toprule
    Name & Embedding time & Target dimension & Ref. & Constraints \\
    \midrule
    Original & $\bigO(\|x\|_0 m)$ & $\bigO(\eps^{-2} \log \frac{1}{\delta})$ & \citep{Johnson:1984:EoLMiaHS} & \\
    Gaussian & $\bigO(\|x\|_0 m)$ & $\bigO(\eps^{-2} \log \frac{1}{\delta})$ & \citep{Har-Peled:2012:ANNTRtCoD} & \\
    Rademacher & $\bigO(\|x\|_0 m)$ & $\bigO(\eps^{-2} \log \frac{1}{\delta})$ & \citep{Arriaga:2006:aAToLRCaRP} & \\
    Achlioptas & $\bigO(\frac{1}{3}\|x\|_0 m)$ & $\bigO(\eps^{-2} \log \frac{1}{\delta})$ & \citep{Achlioptas:2003:DFRPJLwBC} & \\
    DKS & $\bigO(\|x\|_0 \eps^{-1} \log \frac{m}{\delta} \log \frac{1}{\delta})$ & $\bigO(\eps^{-2} \log \frac{1}{\delta})$ & \citep{Dasgupta:2010:aSJLT,Kane:2014:SJLT}& \\
    Block JL & $\bigO(\|x\|_0 \eps^{-1} \log \frac{1}{\delta})$ & $\bigO(\eps^{-2} \log \frac{1}{\delta})$ & \citep{Kane:2014:SJLT} & \\
    Feature Hashing & $\bigO(\|x\|_0)$ & $\bigO(\eps^{-2} \delta^{-1})$ & \citep{Weinberger:2009:FHfLSML,Freksen:2018:FUtHT} & \\
    Feature Hashing & $\bigO(\|x\|_0)$ & $\bigO(\eps^{-2} \log \frac{1}{\delta})$ & \citep{Freksen:2018:FUtHT} & $\nu \ll 1$ \\
    FJLT & $\bigO(d \log d + m \log^2 \frac{1}{\delta})$ & $\bigO(\eps^{-2} \log \frac{1}{\delta})$ & \citep{Ailon:2009:tFJLTaANN} & \\
    BCHJL & $\bigO(d \log m)$ & $\bigO(\eps^{-2} \log \frac{1}{\delta})$ & \citep{Ailon:2009:FDRuRSoDBCHC} & $m = o(d^{1/2})$ \\
    SRHT & $\bigO(d \log d)$ & $\bigO(\eps^{-2} \log |X| \log^4 d)$ & \citep{Ailon:2013:aAOUFJLT,Krahmer:2011:NaIJLEvtRIP} & \\
    SRHT & $\bigO(d \log d)$ & $\bigO(\eps^{-2} \log^3 |X|)$ & \citep{Do:2009:FaEDRuSRM} & $|X| \geq d$ \\
    Toeplitz & $\bigO(d \log m)$ & $\bigO\Bigl($\parbox{3.5cm}{$\eps^{-1} \log^{3/2} |X| \log^{3/2} d \allowbreak + \eps^{-2} \log |X| \log^4 d$}$\Bigr)$ & \citep{Krahmer:2011:NaIJLEvtRIP} & \\
    Toeplitz & $\bigO(d \log m)$ & $\bigO(\eps^{-2} \log^2 \frac{1}{\delta})$ & \citep{Hinrichs:2011:JLLfCM,Vybiral:2011:aVotJLLfCM} & \\
    Toeplitz & $\bigO(d \log m)$ & $\Omega(\eps^{-2} \log^2 \frac{1}{\delta})$ & \citep{Freksen:2020:oUTaCMfJLT} & \\
    GRHD & $\bigO(d \log m)$ & $\bigO(\eps^{-2} \log \frac{1}{\delta})$ & \citep{Bamberger:2017:OFJLEfLDS} & $m = o(d^{1/2})$ \\
    GRHD & $\bigO(d^{2r} \log^4 d)$ & $\bigO(\eps^{-2} \log \frac{1}{\delta})$ & \citep{Bamberger:2017:OFJLEfLDS} & $m = O(d^r)$ \\
    LWTJL & $\bigO(d)$ & $\bigO(\eps^{-2} \log \frac{1}{\delta})$ & \citep{Liberty:2011:DFRPaLWT} & $\nu = \bigO(m^{-1/2}d^{-\delta})$ \\
    KacJL & $\bigO\biggl($ \parbox{6.4cm}{$d \log d \\
    + \min \Bigl\{$ \parbox{4.9cm}{$d \log |X|, \\
    \eps^{-2} \log^2 |X| \log^2 \log |X| \log^3 d$} $\Bigr\}$} $\biggr)$ & $\bigO(\eps^{-2} \log |X|)$ & \citep{Jain:2020:KMJaLaMOFJLT} & \\
    KacJL & $\bigO(d \log d)$ & $\bigO(\eps^{-2} \log |X| \log^2 \log |X| \log^3 d)$ & \citep{Jain:2020:KMJaLaMOFJLT} &
    \\
    \bottomrule
  \end{tabular}
  \caption{The tapestry tabularised.}
  \label{tab:intro:tapestry}
\end{sidewaystable}

\clearpage
\global\pdfpageattr\expandafter{\the\pdfpageattr/Rotate 00}
\pagestyle{plain}


%% file: appendices.tex
\chapter{Deferred Proofs}
\label{cha:deferred-proofs-from-intro}

\section{\texorpdfstring{$k$}{k}-Means Cost is Pairwise Distances}
\label{sec:app:k-means-cost}

Let us first repeat the lemma to remind ourselves what we need to
show.
\begin{replemma}{thm:kmeans:cost_pwd}
  Let $k, d \in \Npos$ and $X_i \subset \R^d$ for
  $i \in \{1, \dotsc, k\}$, then
  \begin{equation*}
    \sum_{i = 1}^{k} \sum_{x \in X_i} \Big\|x - \frac{1}{|X_i|}\sum_{y \in X_i} y \Big\|_2^2
    = \frac{1}{2} \sum_{i = 1}^{k} \frac{1}{|X_i|} \sum_{x, y \in X_i}\|x - y\|_2^2.
  \end{equation*}
\end{replemma}

In order to prove \cref{thm:kmeans:cost_pwd} we will need the
following lemma.
\begin{lemma}
  \label{thm:app:mult_mean_is_zero}
  Let $d \in \Npos$ and $X \subset \R^d$ and define
  $\mu \eqdef \frac{1}{|X|}\sum_{x \in X} x$ as the mean of $X$, then it holds that
  \begin{equation*}
    \sum_{x, y \in X} \dotprod{x - \mu}{y - \mu} = 0.
  \end{equation*}
\end{lemma}
\begin{proof}[Proof of \cref{thm:app:mult_mean_is_zero}]
  The lemma follows from the definition of $\mu$ and the linearity of
  the real inner product.
  \begin{align*}
    \sum_{x, y \in X} \dotprod{x - \mu}{y - \mu}
    &= \sum_{x, y \in X} \big( \dotprod{x}{y} - \dotprod{x}{\mu} - \dotprod{y}{\mu} + \dotprod{\mu}{\mu} \big) \\
    &= \sum_{x, y \in X} \dotprod{x}{y}
      - \sum_{x \in X} 2|X| \dotprod{x}{\mu}
      + |X|^2 \dotprod{\mu}{\mu} \\
    &= \sum_{x, y \in X} \dotprod{x}{y}
      - 2 \sum_{x \in X} \dotprod[\big]{x}{\sum_{y \in X} y}
      + \dotprod[\big]{\sum_{x \in X} x}{\sum_{y \in X} y} \\
    &= \sum_{x, y \in X} \dotprod{x}{y}
      - 2 \sum_{x, y \in X} \dotprod{x}{y}
      + \sum_{x, y \in X} \dotprod{x}{y} \\
    &= 0.
  \end{align*}
\end{proof}

\begin{proof}[Proof of \cref{thm:kmeans:cost_pwd}]
  We will first prove an identity for each partition, so let
  $X_i \subseteq X \subset \R^d$ be any partition of the dataset $X$
  and define $\mu_i \eqdef \frac{1}{|X_i|}\sum_{x \in X_i} x$ as the
  mean of $X_i$.
  \begin{align*}
    \frac{1}{2|X_i|} \sum_{x, y \in X_i} \|x - y \|_2^2
    &= \frac{1}{2|X_i|} \sum_{x, y \in X_i} \|(x - \mu_i) - (y - \mu_i) \|_2^2 \\
    &= \frac{1}{2|X_i|} \sum_{x, y \in X_i} \big(
      \|x - \mu_i\|_2^2 +
      \|y - \mu_i\|_2^2 -
      2 \dotprod{x - \mu_i}{y - \mu_i}
      \big) \\
    &= \sum_{x \in X_i} \|x - \mu_i\|_2^2
      - \frac{1}{2|X_i|} \sum_{x, y \in X_i} 2 \dotprod{x - \mu_i}{y - \mu_i} \\
    &= \sum_{x \in X_i} \|x - \mu_i\|_2^2,
  \end{align*}
  where the last equality holds by
  \cref{thm:app:mult_mean_is_zero}. We now substitute each term in the
  sum in \cref{thm:kmeans:cost_pwd} using the just derived identity:
  \begin{equation*}
    \sum_{i = 1}^{k} \sum_{x \in X_i} \Big\|x - \frac{1}{|X_i|}\sum_{y \in X_i} y \Big\|_2^2
    = \frac{1}{2} \sum_{i = 1}^{k} \frac{1}{|X_i|} \sum_{x, y \in X_i}\|x - y\|_2^2.
  \end{equation*}
\end{proof}

\section{Super Sparse DKS}
\label{sec:app:dks-supersparse}

The tight bounds on the performance of feature hashing presented in
\cref{thm:fh:main} can be extended to tight performance bounds for the
DKS construction. Recall that the DKS construction, parameterised by a
so-called column sparsity $s \in \Npos$, works by first mapping a
vector $x \in \R^d$ to an $x' \in \R^{sd}$ by duplicating each entry
in $x$ $s$ times and then scaling with $1/\sqrt{s}$, before applying
feature hashing to $x'$, as $x'$ has a more palatable
$\ell_\infty / \ell_2$ ratio compared to $x$. The setting for the
extended result is that if we wish to use the DKS construction but we
only need to handle vectors with a small $\|x\|_\infty / \|x\|_2$
ratio, we can choose a column sparsity smaller than the usual
$\Theta(\eps^{-1} \log \frac{1}{\delta} \log\frac{m}{\delta})$ and
still get the Johnson--Lindenstrauss guarantees. This is formalised in
\cref{thm:dks:supersparse}. The two pillars of \cref{thm:fh:main} we
use in the proof of \cref{thm:dks:supersparse} is that the feature
hashing tradeoff is tight and that we can force the DKS construction
to create hard instances for feature hashing.

\begin{repcorollary}{thm:dks:supersparse}
  Let $\nu_{\mathsf{DKS}} \in [1/\sqrt{d}, 1]$ denote the largest
  $\ell_{\infty} / \ell_2$ ratio required, $\nu_{\mathsf{FH}}$ denote
  the $\ell_{\infty} / \ell_2$ constraint for feature hashing as
  defined in \cref{thm:fh:main}, and $s_{\mathsf{DKS}} \in [m]$ as the
  minimum column sparsity such that the DKS construction with that
  sparsity is a JLD for the subset of vectors $x \in \R^d$ that
  satisfy $\|x\|_\infty / \|x\|_2 \leq \nu_{\mathsf{DKS}}$. Then
  \begin{equation}
    \label{eq:dks:supersparse}
    s_{\mathsf{DKS}} = \Theta\Bigl( \frac{\nu_{\mathsf{DKS}}^2}{\nu_{\mathsf{FH}}^2} \Bigr).
  \end{equation}
\end{repcorollary}
The upper bound part of the $\Theta$ in \cref{thm:dks:supersparse}
shows how sparse we can choose the DKS construction to be and still
get Johnson--Lindenstrauss guarantees for the data we care about,
while the lower bound shows that if we choose a sparsity below this
bound, there exists vectors who get distorted too much too often
despite having an $\ell_\infty / \ell_2$ ratio of at most
$\nu_{\mathsf{DKS}}$.

\begin{proof}[Proof of~\cref{thm:dks:supersparse}]
  Let us first prove the upper bound:
  $s_{\mathsf{DKS}} = \bigO\bigl(
  \frac{\nu_{\mathsf{DKS}}^2}{\nu_{\mathsf{FH}}^2} \bigr)$.

  Let
  $s \eqdef \Theta\bigl( \frac{\nu_{\mathsf{DKS}}^2}{\nu_{\mathsf{FH}}^2}
  \bigr) \in [m]$ be the column sparsity, and let $x \in \R^d$ be a
  unit vector with $\|x\|_\infty \leq \nu_{\mathsf{DKS}}$. The goal is
  now to show that a DKS construction with sparsity $s$ can embed $x$
  while preserving its norm within $1 \pm \eps$ with probability at
  least $1 - \delta$ (as defined in \cref{thm:distjl}). Let
  $x' \in \R^{sd}$ be the unit vector constructed by duplicating each
  entry in $x$ $s$ times and scaling with $1/\sqrt{s}$ as in the DKS
  construction. We now have
  \begin{equation}
    \label{eq:dks_cor:xprime_infty_upper}
    \|x'\|_\infty \leq \frac{\nu_{\mathsf{DKS}}}{\sqrt{s}} = \Theta(\nu_{\mathsf{FH}}).
  \end{equation}

  Let $\mathsf{DKS}$ denote the JLD from the DKS construction with
  column sparsity $s$, and let $\mathsf{FH}$ denote the feature
  hashing JLD. Then we can conclude
  \begin{align*}
    \Pr_{f \sim \mathsf{DKS}}\Bigl[ \bigl| \|f(x)\|_2^2 - 1 \bigr| \leq \eps \Bigr]
    &= \Pr_{g \sim \mathsf{FH}}\Bigl[ \bigl| \|g(x')\|_2^2 - 1 \bigr| \leq \eps \Bigr]
      \geq 1 - \delta,
  \end{align*}
  where the inequality is implied by \cref{eq:dks_cor:xprime_infty_upper,thm:fh:main}.

  Now let us prove the lower bound:
  $s_{\mathsf{DKS}} = \Omega\bigl(
  \frac{\nu_{\mathsf{DKS}}^2}{\nu_{\mathsf{FH}}^2} \bigr)$.

  Let
  $s \eqdef o\bigl( \frac{\nu_{\mathsf{DKS}}^2}{\nu_{\mathsf{FH}}^2}
  \bigr)$, and let
  $x = (\nu_{\mathsf{DKS}}, \dotsc, \nu_{\mathsf{DKS}}, 0, \dotsc,
  0)^\T \in \R^d$ be a unit vector. We now wish to show that a DKS
  construction with sparsity $s$ will preserve the norm of $x$ to
  within $1 \pm \eps$ with probability strictly less than
  $1 - \delta$. As before, define $x' \in \R^{sd}$ as the unit vector
  the DKS construction computes when duplicating every entry in $x$
  $s$ times and scaling with $1/\sqrt{s}$. This gives
  \begin{equation}
    \label{eq:dks_cor:xprime_infty_lower}
    \|x'\|_\infty = \frac{\nu_{\mathsf{DKS}}}{\sqrt{s}} = \omega(\nu_{\mathsf{FH}}).
  \end{equation}

  Finally, let $\mathsf{DKS}$ denote the JLD from the DKS construction with
  column sparsity $s$, and let $\mathsf{FH}$ denote the feature
  hashing JLD. Then we can conclude
  \begin{align*}
    \Pr_{f \sim \mathsf{DKS}}\Bigl[ \bigl| \|f(x)\|_2^2 - 1 \bigr| \leq \eps \Bigr]
    &= \Pr_{g \sim \mathsf{FH}}\Bigl[ \bigl| \|g(x')\|_2^2 - 1 \bigr| \leq \eps \Bigr]
      < 1 - \delta,
  \end{align*}
  where the inequality is implied by
  \cref{eq:dks_cor:xprime_infty_lower,thm:fh:main}, and the fact
  that $x'$ has the shape of an asymptotically worst case instance for
  feature hashing.
\end{proof}

\section{LWTJL Fails for Too Sparse Vectors}
\label{sec:app:lwt-nu-bound}

\begin{repproposition}{thm:lwt:nu_upper_bound}
  For any seed matrix define $\mathsf{LWT}$ as the LWTJL distribution
  seeded with that matrix. Then for all $\delta \in (0, 1)$, there
  exists a vector $x \in \CC^d$ (or $x \in \R^d$, if the seed matrix
  is a real matrix) satisfying
  $\|x\|_\infty / \|x\|_2 = \Theta(\log^{-1/2} \frac{1}{\delta})$ such
  that
  \begin{equation*}
    \Pr_{f \sim \mathsf{LWT}}[f(x) = \zerovector] > \delta.
  \end{equation*}
\end{repproposition}
\begin{proof}
  The main idea is to construct the vector $x$ out of segments that
  are orthogonal to the seed matrix with some probability, and then
  show that $x$ is orthogonal to all copies of the seed matrix
  simultaneously with probability larger than $\delta$.

  Let $r, c \in \Npos$ be constants and $A_1 \in \CC^{r \times c}$ be a
  seed matrix. Let $d$ be the source dimension of the LWTJL
  construction, $D \in \{-1, 0, 1\}^{d \times d}$ be the random
  diagonal matrix with i.i.d. Rademachers, $l \in \Npos$ such that
  $c^l = d$, and $A_l \in \CC^{r^l \times c^l}$ be the the LWT, i.e.
  $A_l \eqdef A_1^{\otimes l}$. Since $r < c$ there exists a
  nontrivial vector $z \in \CC^c \setminus \{\zerovector\}$ that is orthogonal
  to all $r$ rows of $A_1$ and $\|z\|_\infty = \Theta(1)$. Now define
  $x \in \CC^d$ as $k \in \Npos$ copies of $z$ followed by a padding of
  $0$s, where
  $k = \lfloor \frac{1}{c} \lg \frac{1}{\delta} - 1\rfloor$. Note that
  if the seed matrix is real, we can choose $z$ and therefore $x$ to
  be real as well.

  The first thing to note is that
  \begin{equation*}
    \|x\|_0 \leq ck < \lg \frac{1}{\delta},
  \end{equation*}
  which implies that
  \begin{equation*}
    \Pr_D[Dx = x] = 2^{-\|x\|_0} > \delta.
  \end{equation*}

  Secondly, due to the Kronecker structure of $A_l$ and the fact that
  $z$ is orthogonal to the rows of $A_1$, we have
  \begin{equation*}
    Ax = \zerovector.
  \end{equation*}
  Taken together, we can conclude
  \begin{equation*}
    \Pr_{f \sim \mathsf{LWT}}[f(x) = \zerovector] \geq \Pr_D[A_lDx = \zerovector] \geq \Pr_D[Dx = x] > \delta.
  \end{equation*}

  Now we just need to show that
  $\|x\|_\infty / \|x\|_2 = \Theta(\log^{-1/2} \frac{1}{\delta})$.
  Since $c$ is a constant and $x$ is consists of
  $k = \Theta(\log \frac{1}{\delta})$ copies of $z$ followed by
  zeroes,
  \begin{alignat*}{3}
    \|x\|_\infty &= \|z\|_\infty &&= \Theta(1), \\
    \|z\|_2 &= \Theta(1), \\
    \|x\|_2 &= \sqrt{k} \|z\|_2 &&= \Theta\bigl(\sqrt{\log \frac{1}{\delta}}\bigr),
  \end{alignat*}
  which implies the claimed ratio,
  \begin{equation*}
    \frac{\|x\|_\infty}{\|x\|_2} = \Theta(\log^{-1/2} \frac{1}{\delta}).
  \end{equation*}
\end{proof}

The following corollary is just a restatement of
\cref{thm:lwt:nu_upper_bound} in terms of \cref{thm:distjl}, and the
proof therefore follows immediately from
\cref{thm:lwt:nu_upper_bound}.
\begin{corollary}
  \label{thm:lwt:nu_upper_bound:corollary}
  For every $m, d, \in \Npos$, and $\delta, \eps \in (0, 1)$, and LWTJL
  distribution $\mathsf{LWT}$ over $f \colon \K^d \to \K^m$, where
  $\K \in \{\R, \CC\}$ and $m < d$ there exists a vector $x \in \K^d$
  with $\|x\|_\infty / \|x\|_2 = \Theta(\log^{-1/2} \frac{1}{\delta})$
  such that
  \begin{equation*}
    \Pr_{f \sim \mathsf{LWT}}\Bigl[ \, \bigl| \,\|f(x)\|_2^2 - \|x\|_2^2\, \bigr| \leq \eps \|x\|_2^2 \,\Bigr] < 1 - \delta .
  \end{equation*}
\end{corollary}
